\documentclass[10pt,journal]{IEEEtran}
\usepackage{cite}
\usepackage{amsmath,amssymb,amsfonts}
\usepackage{graphicx}
\usepackage{textcomp}
\usepackage{float}
\usepackage{mathtools}
\usepackage{algorithm}
\usepackage{algorithmicx}
\usepackage{algpseudocode}
\usepackage{amsthm}
\usepackage[nocomma]{optidef}
\usepackage[justification=centering]{caption}
\usepackage{enumitem}

\makeatletter
\def\endthebibliography{%
	\def\@noitemerr{\@latex@warning{Empty `thebibliography' environment}}%
	\endlist
}
\makeatother

\begin{document}
\algnewcommand{\algorithmicgoto}{\textbf{go to}}%
\algnewcommand{\Goto}[0]{\algorithmicgoto~}%

\newtheorem{thm}{Theorem}
\newtheorem{lem}{Lemma}
\newtheorem{prop}[thm]{Proposition}
\newtheorem{cor}{Corollary}
\newtheorem{conj}{Conjecture}[section]
\newtheorem{defn}{Definition}
\newtheorem{exmp}{Example}[section]
\newtheorem{rem}{Remark}
\newcommand{\norm}[1]{\left\lVert#1\right\rVert}
\DeclarePairedDelimiter\abs{\lvert}{\rvert}%

\title{Energy-aware Resource Management for Federated Learning in Multi-access Edge Computing Systems}
\author{\uppercase{Chit Wutyee Zaw},
\uppercase{Shashi Raj Pandey} \IEEEmembership{Student Member, IEEE}, \uppercase{Kitae Kim}, and \uppercase{Choong Seon Hong}
\IEEEmembership{Senior Member, IEEE}
\thanks{Chit Wutyee Zaw, Shashi Raj Pandey, Kitae Kim, and Choong Seon Hong  are with the Department of Computer Science and Engineering, Kyung Hee University,  Yongin-si, Gyeonggi-do 17104, Rep. of Korea.
E-mail: {\{cwyzaw, shashiraj, glideslope, cshong\}@khu.ac.kr}}}

\maketitle

\begin{abstract}
In Federated Learning (FL), a global statistical model is developed by encouraging mobile users to perform the model training on their local data and aggregating the output local model parameters in an iterative manner. However, due to limited energy and computation capability at the mobile devices, the performance of the model training is always at stake to meet the objective of local energy minimization. In this regard, Multi-access Edge Computing (MEC)-enabled FL addresses the tradeoff between the model performance and the energy consumption of the mobile devices by allowing users to offload a portion of their local dataset to an edge server for the model training. Since the edge server has high computation capability, the time consumption of the model training at the edge server is insignificant. However, the time consumption for dataset offloading from mobile users to the edge server has a significant impact on the total time consumed to complete a single round of FL process. Thus, resource management in MEC-enabled FL is challenging, where the objective is to reduce the total time consumption while saving the energy consumption of the mobile devices. In this paper, we formulate an energy-aware resource management for MEC-enabled FL in which the model training loss and the total time consumption are jointly minimized, while considering the energy limitation of mobile devices.  In addition, we recast the formulated problem as a Generalized Nash Equilibrium Problem (GNEP) to capture the coupling constraints between the radio resource management and dataset offloading. To that end, we analyze the impact of the dataset offloading and computing resource allocation on the model training loss, time, and the energy consumption. Finally, we present the convergence analysis of the proposed solution, and evaluate its performance against the traditional FL approach. Simulation results demonstrate the efficacy of our proposed solution approach.

\end{abstract}


\section{Introduction}\label{sec:introduction}
Federated Learning (FL) builds a statistical model by allowing mobile users to train local models on datasets residing at their mobile devices \cite{mcmahan2017communication}. The users only share the trained local model parameters to a central server for model aggregation; thus, the local datasets' privacy is preserved. In recent years, several works study FL over wireless networks \cite{shashi_twc, naderializadeh2020communication, prof_nguyen_fl, yang2019energy,chen_joint_learning_communication,chen2020convergence, wadu2020federated}. These works are motivated by the possibility of leveraging existing cellular infrastructure for offering learning services to the users via distributed model training approach, such as FL \cite{niknam2020federated}. However, most of the works \cite{naderializadeh2020communication, prof_nguyen_fl, yang2019energy,chen_joint_learning_communication,chen2020convergence, wadu2020federated} highlight the implication of wireless resource optimization, convergence analysis, and training-time minimization when performing distributed model training over dynamic wireless conditions. Moreover, there are several other overlooked challenges and open problems for the direct implementation of FL over wireless networks \cite{fl_challenges, kairouz2019advances}. On the one hand, the model training's performance is significantly influenced by local datasets and computing resources used for the training. On the other hand, the subset of mobile devices selected in each round of model training affects the time required to reach a global model's desired accuracy level. This situation gets exacerbated when we have dynamic wireless conditions.

The trade-off between the model performance, energy and time consumption can be resolved by enabling Multi-access Edge Computing (MEC) in FL \cite{fl_mec_survey}. In particular, MEC brings the high computing servers closer to the mobile users so that users with low computing and energy capability are able to offload their latency and computing-intensive tasks to the edge server \cite{AbbasFeb.2018,mec_survey_architecture,mec_survey}. Therefore, mobile users are able to offload a selected portion of local dataset to the edge server where a statistical model is trained by the edge server simultaneously with several mobile devices in hands \cite{hybrid_fl,hiding_crowd}. Even though FL is intended for the privacy preserving application, a portion of local dataset which are not privacy-sensitive can be offloaded to the MEC for further computation. Then, the MEC server can perform the model training on all the datasets offloaded by the mobile users, simultaneously, and perform averaging of local model parameters and the obtained model to build a single global model. Besides, the users can determine the offloaded data samples based on the freshness of the collected data. Thus, this approach is more practical as it should be up to the users to decide the kind of data they want to share and further improve the model performance.

Moreover, the performance of the global model in FL is highly affected by the heterogeneity in computing resources of the mobile device for training the local model. Besides, due to the energy limitation of the mobile devices, the user may use less amount of local dataset and computing resource for the model training, which would result in lower model performance. Thus, the trade-off between the energy consumption of mobile devices and performance of the training model is required to be addressed in FL. In this regard, the edge server is a powerful computing device; hence, the time and energy consumption of the model training at the edge server is negligible. Therefore, it is intuitive to leverage the MEC infrastructure for sharing computation burden of resource constrained mobile devices during the model training process in FL. By allowing the mobile users to offload a portion of their local datasets to the edge server, the performance of the global model can be preserved while saving the energy consumption of the mobile devices.

Inline with this idea, the works in \cite{hiding_crowd} and \cite{hybrid_fl} proposed the local data sharing mechanism for FL. In \cite{hybrid_fl}, the authors mitigated the non-i.i.d. data problem by allowing a limited number of users to upload their local data to a server; and thus, the server trains a model on the uploaded data to support the FL process during model aggregation. Authors in \cite{hiding_crowd} proposed a distributed data augmentation algorithm in which users share a fraction of their local dataset to confront the lack of on-device data samples. Similar to these approaches \cite{hiding_crowd,hybrid_fl}, we give users the ultimate power to decide the dataset offloading. On the other hand, as a ML developer, the global model gets benefited with our proposed scheme, wherein we balance between the high accuracy obtained in a centralized setting and the distributed privacy preserving model training framework, such as FL. The proposed mechanism is practical and can be applied to real-time applications such as autonomous driving and mobile surveillance, where the privacy of the data collected from the devices is not a major concern.

 In summary, we raise two overlooked yet fundamentally coupled research questions here:
\begin{itemize}
\item \textit{How to involve more number of mobile devices, having a moderate computational capacity and reasonable privacy concerns, in the FL training process?}
\item \textit{How to perform an efficient resource optimization while ensuring the model performance?}
\end{itemize}

In this paper, we propose a MEC-enabled FL model to address the tradeoff between the training model's performance, total time, and energy consumption of mobile devices. The joint model learning and resource management problem is challenging due to the coupling among the offloading decision and resource management. Thus, Generalized Nash Equilibrium Game is formulated for the dataset offloading and uplink radio resource management to minimize the total time taken for one global iteration. The mobile users' energy limitation is considered in the local computing resource management problem where the mobile users have a moderate computational capacity and reasonable privacy concerns. The energy-aware resource management algorithm for the MEC-enabled FL is proposed in which the model training and resource management problems are solved alternatively.

\subsection{Related Works}
\subsubsection{Resource Management in FL}
The wireless resource management has been an interesting topic in FL. Author in \cite{naderializadeh2020communication} analyzed the communication latency for decentralized learning over wireless networks, where each node is allowed to communicate with its own neighbors. The optimization model is proposed in \cite{prof_nguyen_fl} for FL over wireless networks, where the energy and time consumption are jointly optimized by power allocation, local computing resource, and model accuracy. FL over wireless communication networks is studied in \cite{yang2019energy,chen_joint_learning_communication,chen2020convergence} in which the authors discussed the joint optimization of the model training and wireless resource allocation. The channel uncertainty is considered in \cite{wadu2020federated} where the joint user scheduling and resource block allocation is performed so as to minimize the loss of FL accuracy. The cost and learning loss are jointly minimized in \cite{closely_related} by selecting mobile users who are participating in FL. The selected users are allowed to determine the amount of data samples used for the model training. Two level aggregation for FL is proposed in \cite{cost_efficient_fl} in which an intermediate model aggregation can be performed at the edge server where the final model aggregation is performed at the cloud server.

\subsubsection{Resource Management in MEC}
The joint optimization of radio and computing resource management in MEC has been studied thoroughly in previous works. Authors in \cite{related_completion_time} proposed a two step optimization for radio and computing resource allocation so as to minimize the total processing time. A multi-cell MEC is considered in \cite{related_joint_multicell} in which the radio and computing resources are jointly optimized to save the energy consumption of the mobile users where the latency limit of the task offloading is considered. The queueing model for resource allocation is studied in \cite{related_latency,related_provision,related_stochastic} in which the stability of the queues is required to be satisfied in task offloading and resource allocation.

\subsubsection{Resource Management with Generalized Nash Equilibrium Problem (GNEP)}
GNEP is a promising technique to handle the strong coupling of optimization variables in resource allocation problems where both the objective and strategy sets of players are dependent on each other. The properties, existence of Generalized Nash Equilibrium (GNE), and solution algorithms are studied in \cite{gnep_concept}. GNEP for service provisioning problem is proposed in \cite{gnep_multi_cloud,gnep_service_provisioning,ardagna2012generalized} to model the multi-cloud systems among multiple service providers. GNEP for the task offloading in MEC is proposed in \cite{gnep_mec} in which the total time consumption is minimized by the offloading decision. The joint radio and computing resource management for MEC is formulated as a GNEP in \cite{our_work_bigcomp,our_work_icc} in which the authors proposed a penalty-based resource management algorithm to find a GNE.

\subsection{Our Contributions}
In this paper, an energy-aware resource management problem is formulated for the MEC-enabled FL model.
Our contributions are as follows:
\begin{itemize}
	\item We propose a MEC-enabled FL in which mobile users are allowed to offload a portion of their local datasets to the edge server. The proposed MEC-enabled FL reflects a practical scenario where the users have moderate computational capabilities and reasonable privacy concerns. The mobile users can determine the offloaded data samples depending on the freshness or privacy of the generated data samples. Moreover, the proposed MEC-enabled FL model addresses the tradeoff among the performance of learning model and energy consumption of the mobile devices. 
	\item The energy-aware resource management problem is formulated for the proposed MEC-enabled FL model. The learning model, dataset offloading, local computing, and uplink radio resources management are jointly optimized to minimize the training loss and time consumption for one global round, ensuring the energy constraints of the mobile devices.
	\item The uplink radio resource management of the edge server and the dataset offloading of the mobile devices are formulated as a GNEP to focus on the coupling among the resource management. The optimal solution of the dataset offloading and the resource management is derived, where the time consumption of the local and edge model is adjusted.
	\item Extensive simulations are performed to compare the performance of the proposed MEC-enabled FL and traditional FL in terms of the learning model, time, and energy consumption for the dataset offloading, and local computing resource management. In addition, we analyze and validate these performance metrics of the proposed algorithm on the cell-center and cell-edge user is analyzed in which the system heterogeneity of the mobile users is considered.
\end{itemize}
\begin{figure}[t!] 
	\centering
	\includegraphics[keepaspectratio=true,scale=0.3]{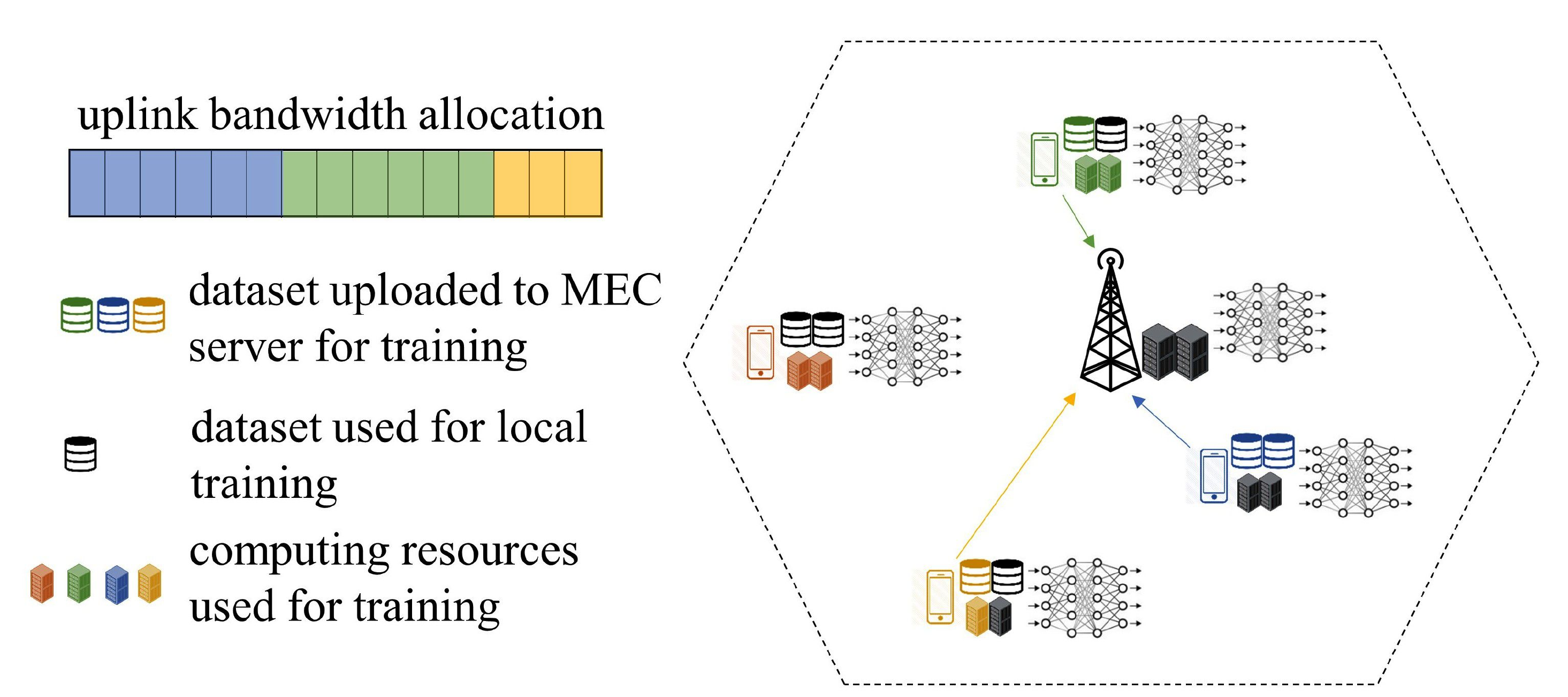}
	\caption{Resource management for MEC-enabled federated learning.}
	\label{fig:system_model}
\end{figure}
The rest of the paper is organized as follows. The system model is presented in Section \ref{sec:system_model} in which the communication and learning model for the MEC-enabled FL model are proposed. The energy aware resource management problem is formulated in Section \ref{sec:problem_formulation} by considering the energy limit of the mobile devices to minimize the total time consumption of one global iteration. The energy aware resource management algorithm for MEC-enabled FL is proposed in Section \ref{sec:solution}. In addition,  the performance of the proposed model is compared with the traditional FL in Section \ref{sec:simulation_results}. The paper is concluded in Section \ref{sec:conclusion}.

\section{System Model}\label{sec:system_model}
A single-cell MEC system is considered in this paper where an edge server is deployed at the access point which is utilized for training a statistical model simultaneously with the mobile devies. The energy consumption for the model training at the mobile users can be reduced by offloading the portion of their datasets to the edge server for the training. In the proposed MEC-enabled FL consists of an edge server, and a set of mobile users, $ \mathcal{I} \in \{ 1, 2, \cdots, I \} $ where user $ i $ has the local dataset $ \mathcal{D}_i $ to train a local model. The proposed MEC-enabled FL system is shown in Fig. \ref{fig:system_model} where the mobile users are allowed to offload $ 0 \leq \delta_i \leq 1 $ portion of the local dataset to the edge server, while the remaining $ (1 - \delta_i) $ portion of the dataset is used for the local model training. Depending on the energy level of the mobile users, the computing resource used for the local training is managed in order to minimize the training loss and time consumption of the model training. The dataset offloading and computing resource allocation is determined by the mobile users individually while the edge server controls the radio resource management for the dataset offloading and weight transmission.

An illustration of the proposed MEC-enabled FL model is shown in Fig. \ref{fig:system_flow} in which the time consumption at each stage is defined. Moreover, the synchronous update model for FL is considered in this paper. The mobile devices participating in FL train their local model with datasets residing at the mobile devices and transmit the weights of the model to the edge server in the traditional FL approach. In our proposed MEC-enabled FL model, the mobile devices simultaneously offload a portion of their datasets to the edge server and train their local model with the remaining portion. After all the offloaded datasets are received by the edge server, the edge training is performed. The model aggregation is carried out after the edge training and the weight transmission of all mobile devices.

\begin{figure}[t!]  
	\centering
	\includegraphics[keepaspectratio=true,scale=0.3]{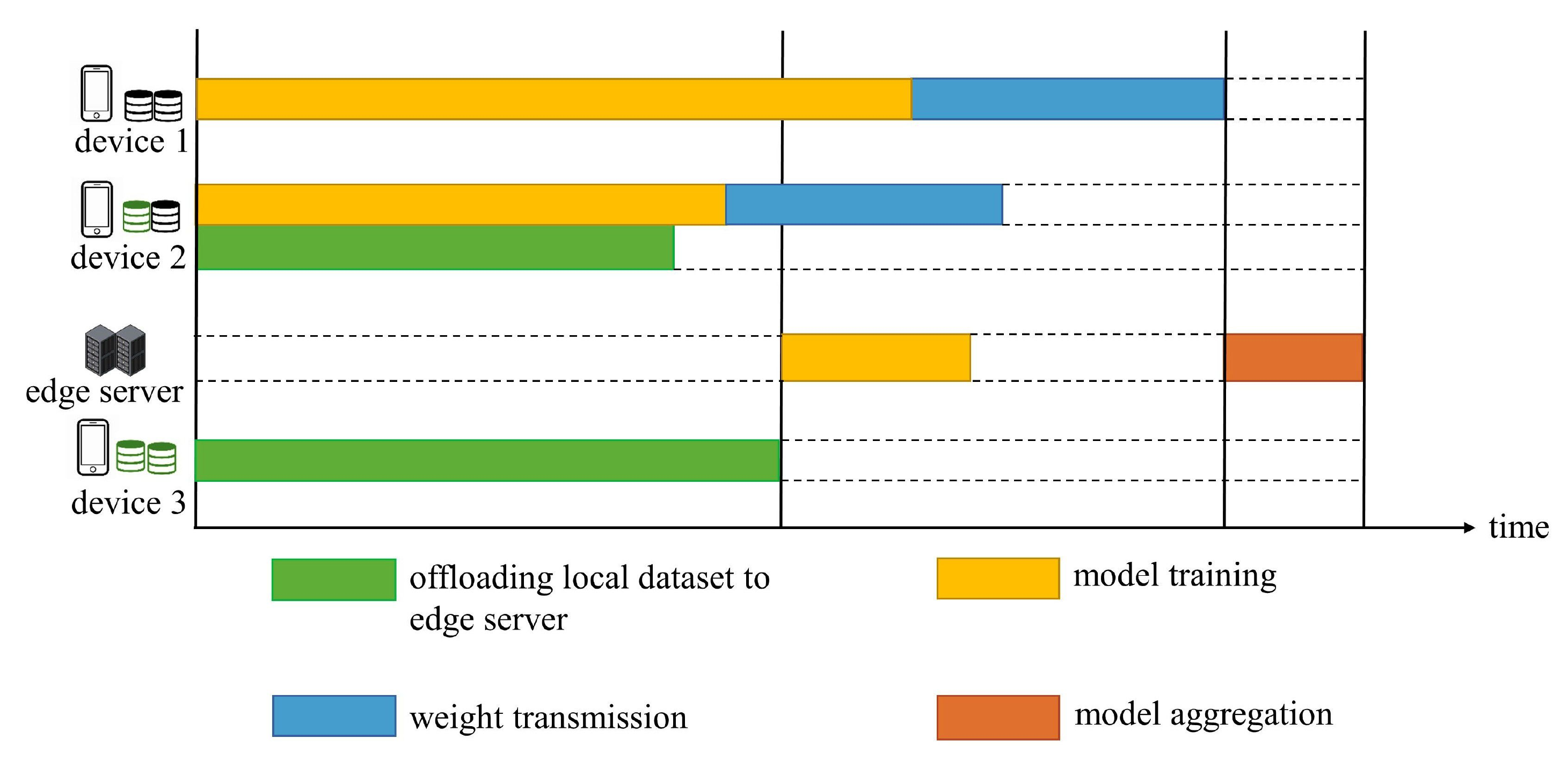}
	\caption{An illustration of MEC-enabled federated learning model.}
	\label{fig:system_flow}
\end{figure}

\subsection{Communication Model}\label{subsec:communication_model}
In this paper, we consider the Orthogonal Frequency Division Multiple Access (OFDMA) for data transmission, i.e., for the dataset offloading and weight transmission. The size of offloaded dataset varies across the mobile users depending on their channel condition and energy level, whereas the size of weight vectors is the same for all mobile users. Thus, the radio resource management for the uplink transmission is performed twice for the dataset offloading and weight uploading. These two transmissions are not performed simultaneously. The fraction of bandwidth allocated to user $ i $ for dataset offloading is denoted as $ \tilde{\omega}_i $, while $ \bar{\omega}_i $ is for the weight uploading. Thus, the achievable data rate of user $ i $ for the dataset offloading is defined as

\begin{equation}
R^{\text{off}}_i = \tilde{\omega}_i \omega \log_2 \left( 1 + \frac{p_i g_i}{n_0} \right).
\end{equation}

In addition, the achievable data rate of user $ i $ in uploading the weight vector is

\begin{equation}
R^{\text{upload}}_i = \bar{\omega}_i \omega \log_2 \left( 1 + \frac{p_i g_i}{n_0} \right),
\end{equation}
where $ \omega $ is the total available bandwidth of the access point for the uplink transmission, $ p_i $  is the transmit power of user $ i $, $ g_i $ is the uplink channel gain of user $ i $, and $ n_0 $ is the additive white Gaussian noise. For simplicity, let $ R_i $ be $ \omega \log_2 \left( 1 + \frac{p_i g_i}{n_0} \right) $. Thus, the achievable data rate for the dataset offloading and weight uploading is denoted by $ \tilde{\omega}_i R_i $ and $ \bar{\omega}_i R_i $ respectively.

\subsection{Federated Learning Model}\label{subsec:learning_model}
In traditional FL, a statistical model is learned by allowing users to train a local model on the dataset residing at their mobile devices. In order to achieve the higher model performance, the objective of user $ i $ is to minimize the training loss by optimizing the weight parameter $ \mathbf{w}_i $ with respect to its local dataset $ \mathcal{D}_i $ as follows:

\begin{equation*}\label{equ:loss_function_traditionalFL}
\underset{\mathbf{w}_i \in \mathbb{R}^n}{\text{minimize}}
\sum_{j \in \mathcal{D}_i} l(\mathbf{w}_i, \mathbf{x}_j, y_j),
\end{equation*}
where $ \mathbf{x}_j \in \mathbb{R}^n $ and $ y_j \in \mathbb{R} $ are the features vector and label of the data sample $ j \in \mathcal{D}_i $, $ n $ is the dimensions of the features vector. In this work, the Mean Squared Error (MSE) is used for calculating the loss function where the logistic regression is implemented for the model learning. After the local training is done at user $ i $, the weight vector $ \mathbf{w}_i $ is sent to the edge server for the model aggregation, where a global model is developed as defined in \cite{mcmahan2017communication} which is as follows:

\begin{equation*}\label{equ:weight_aggregation}
\bar{\mathbf{w}} = \frac{\sum_{i \in \mathcal{I}} |\mathcal{D}_i| \mathbf{w}_i}{\sum_{i \in \mathcal{I}} |\mathcal{D}_i|}.
\end{equation*}
The final model is derived by the contribution of mobile users which is defined as the proportion of the size of their local dataset to the total dataset.

The performance of the global model depends not only on the local dataset but also the computing resources used for the local training. A local model at user $ i $ is trained by updating the $ \mathbf{w}_i $ at multiple iterations according to stochastic gradient descent approach. User $ i $ could save its energy consumption by stopping the training after a few iterations. In order to preserve the performance of the final global model and energy consumption of the mobile devices, users can offload a portion of their dataset to the edge server. The proposed MEC-enabled FL is described in the following sections.

\subsection{Local Training Model}\label{subsec:local_training_model}
In the proposed MEC-enabled FL model, user $ i $ is allowed to offload a portion $ \delta_i $ of its local dataset $ \mathcal{D}_i $ to the edge server, while the remaining $ (1 - \delta_i) $ portion of the local dataset is used in training a model on the mobile devices. The objective of user $ i $ is to optimize a weight vector, $ \mathbf{w}_i $, by minimizing the training loss which is defined as follows:

\begin{equation}\label{equ:loss_function_mecFL}
\underset{\mathbf{w}_i \in \mathbb{R}^n}{\text{minimize}}
\sum_{j \in \bar{\mathcal{D}}_i} l(\mathbf{w}_i, \mathbf{x}_j, y_j),
\end{equation}
where $ \bar{\mathcal{D}}_i $ is the dataset to train the local model in which samples are chosen randomly from $ \mathcal{D}_i $ and $ |\bar{\mathcal{D}}_i| = (1 - \delta_i) |\mathcal{D}_i| $. After the local model is trained on the mobile device, user $ i $ uploads the weight vector, $ \mathbf{w}_i $, to the edge server which describes the local model. Thus, user $ i $ needs to perform two independent operations: i) the local training and ii) weight transmission to the edge server. Therefore, the time consumption of user $ i $ to execute the two stages is defined as

\begin{equation}\label{equ:local_time}
t_i^{\text{local}} = \frac{(1 - \delta_i) f(|\mathcal{D}_i|) \tau}{\gamma_i \Gamma_i} + \frac{f(|\mathbf{w}_i|)}{\bar{\omega}_i R_i},
\end{equation}
where $ f(|\mathcal{D}_i|) $ is a linear function of $ |\mathcal{D}_i| $ which defines the size of user $ i $'s local dataset in bytes, $ \tau $ is the number of CPU cycles required for one byte in the model training, $ \gamma_i $ is the fraction of CPU resources used for the model training, $ \Gamma_i $ is the total available CPU resources at user $ i $, and $ f(|\mathbf{w}_i|) $ is a linear function of the weight vector $ \mathbf{w}_i $ which defines the size of the weight vector in bytes.
The local energy consumption of user $ i $ for the local model training and weight transmission to the edge server is calculated as
\begin{equation}\label{equ:local_energy}
e_i^{\text{local}} = \psi (1 - \delta_i) f(|\mathcal{D}_i|) \tau (\gamma_i \Gamma_i)^2 + p_i \left( \frac{f(|\mathbf{w}_i|)}{\bar{\omega}_i R_i} \right),
\end{equation}
where $ \psi $ is the chip capacitance related to the CPU of the mobile device as defined in \cite{energy_local}.

\subsection{Edge Training Model}\label{subsec:edge_training_model}
In order to preserve the local energy and the performance of the final model, user $ i $ is allowed to offload $ \delta_i $ portion of its local dataset $ \mathcal{D}_i $ to the edge server, where the model training is performed on the offloaded dataset in which the weight vector $ \mathbf{w}_E $ is optimized to minimize the training loss as follows:
\begin{equation}\label{equ:loss_function_edge}
\underset{\mathbf{w}_E \in \mathbb{R}^n}{\text{minimize}}
\sum_{j \in \tilde{\mathcal{D}}_E} l(\mathbf{w}_E, \mathbf{x}_j, y_j),
\end{equation}
where $ \tilde{\mathcal{D}}_E = \cup_{i \in \mathcal{I}} \tilde{\mathcal{D}}_i $, $ \tilde{\mathcal{D}}_i $ is the offloaded dataset of user $ i $ in which the data samples are chosen randomly from the local dataset $ \mathcal{D}_i $ such that $ |\tilde{\mathcal{D}}_i| = \delta_i |\mathcal{D}_i| $, $ \tilde{\mathcal{D}}_i \cup \bar{\mathcal{D}} = \mathcal{D}_i $, and $ \tilde{\mathcal{D}}_i \cap \bar{\mathcal{D}}_i = \emptyset $.
The model training at the edge server involves two stages as well which are the dataset offloading of the mobile users and the weight optimization. Thus, the time consumption of the edge training is defined as follows:
\begin{equation}\label{equ:edge_time}
t^{\text{edge}} = \max_{i \in \mathcal{I}} \left \{ \frac{\delta_i f(|\mathcal{D}_i|)}{\tilde{\omega}_i R_i} \right \} + \frac{\sum_{i \in \mathcal{I}} \delta_i f(|\mathcal{D}_i|) \tau}{\Gamma_E},
\end{equation}
where $ \Gamma_E $ is the available CPU resources of the edge server.
The energy consumption of user $ i $ in the dataset offloading is calculated as follows:
\begin{equation}\label{equ:off_energy}
e_i^{\text{off}} = p_i \left( \frac{\delta_i f(|\mathcal{D}_i|)}{\tilde{\omega} R_i} \right).
\end{equation}

Once the model training at the edge server is executed and the weight transmission from the mobile users is completed, the final model aggregation is performed. Following the analysis of \cite{mcmahan2017communication}, we define the final model aggregation as follows:
\begin{equation}\label{equ:model_agg_edge}
\bar{\mathbf{w}} = \frac{\sum_{i \in \mathcal{I}} |\bar{\mathcal{D}}_i| \mathbf{w}_i + |\tilde{\mathcal{D}}_E| \mathbf{w}_E}{\sum_{i \in \mathcal{I}} |\mathcal{D}_i|},
\end{equation}
where the contribution of the weights from the mobile users and the edge server to the final model is proportional to their data samples used in the model training.

\section{Problem Formulation}\label{sec:problem_formulation}
In this section, the energy aware resource management problem for the MEC-enabled FL is formulated, where the training loss and the time consumption for one communication is jointly minimized while considering the energy level of the mobile devices. It is crucial to minimize the time taken for one communication round in the proposed MEC-enabled FL because the time taken for dataset offloading influences the total time consumption which can be significantly higher than the traditional FL model. Since both the edge server, and the mobile user $ i $ gets involve in the model training for the proposed MEC-enabled FL model, the total time taken for one communication round is defined as follows:
\begin{equation}\label{equ:total_time}
t^{\text{total}} = \max \left \{ \max_{i \in \mathcal{I}} \left \{ t_i^{\text{local}} \right \}, t^{\text{edge}} \right \},
\end{equation}
where the synchronous model update is considered.
In addition, the total energy consumption of user $ i $ in the proposed MEC-enabled FL is defined as follows:
\begin{equation}\label{equ:total_energy}
e_i^{\text{total}} = e_i^{\text{local}} + e_i^{\text{off}}.
\end{equation}

Thus, the energy aware resource management problem for the MEC-enabled FL is formulated in which the training loss and the total time consumption are jointly minimized by guaranteeing the energy limit of the mobile devices.

The optimization problem of the edge server, where the uplink radio resources are managed so as to jointly minimize the edge training loss and total time consumption is defined as follows:
\begin{equation}\label{problem_edge_server}
\begin{aligned}
& \underset{\mathbf{w}_E \in \mathbb{R}^n, [\boldsymbol{\tilde{\omega}}, \boldsymbol{\bar{\omega}}] \in \mathcal{S}_E}{\text{minimize}}
& & \alpha_1 \left[ \sum_{j \in \bar{\mathcal{D}}_E} l(\mathbf{w}_E, x_j, y_j) \right] + \alpha_2 \, t^{\text{total}} \\
& \text{subject to}
& & e_i^{\text{total}} \leq \Delta e_i, \forall i \in \mathcal{I},
\end{aligned}
\end{equation}
where $ \alpha_1 $ and $ \alpha_2 $ are, respectively, the scaling parameters for the training loss and total time. 
The set $ \mathcal{S}_E $ is defined as
\begin{equation}\label{equ:edge_set}
\mathcal{S}_E = \left \{ [\boldsymbol{\tilde{\omega}}, \boldsymbol{\bar{\omega}}] | \sum_{i \in \mathcal{I}} \tilde{\omega}_i \leq 1, \sum_{i \in \mathcal{I}} \bar{\omega}_i \leq 1, \boldsymbol{\tilde{\omega}}, \boldsymbol{\bar{\omega}} \geq 0 \right \},
\end{equation}
where $ \boldsymbol{\tilde{\omega}} = \left[ \tilde{\omega}_i \right]^T_{i \in \mathcal{I}} $, $ \boldsymbol{\bar{\omega}} = \left[ \bar{\omega}_i \right]^T_{i \in \mathcal{I}} $. $ \mathcal{S}_E $ defines the limit of the uplink radio resource allocation where the sum of the fraction of uplink bandwidth allocated to the mobile users must not exceed 1 for both dataset offloading and weight transmission, and the resource allocation must be non-negative.

The dataset offloading and computing resources optimization of user $ i $ to jointly optimize the local training loss and the total time consumption is as follows:
\begin{equation}\label{problem_local}
\begin{aligned}
& \underset{\mathbf{w}_i \in \mathbb{R}^n, [\delta_i, \gamma_i] \in \mathcal{S}_i}{\text{minimize}}
& & \alpha_1 \left[ \sum_{j \in \bar{\mathcal{D}}_i} l(\mathbf{w}_i, x_j, y_j) \right] + \alpha_2 \, t^{\text{total}} \\
& \text{subject to}
& & e_i^{\text{total}} \leq \Delta e_i.
\end{aligned}
\end{equation}
The set, $ \mathcal{S}_i $ is defined as
\begin{equation}\label{local_set}
\mathcal{S}_i = \left \{ [\delta_i, \gamma_i] | 0 \leq \delta_i \leq 1, 0 \leq \gamma_i \leq 1 \right \},
\end{equation}
where $ \mathcal{S}_i $ defines the boundaries for the data offloading variable $ \delta_i $, and computing resource allocation $ \gamma_i $, to ensure that $ \delta_i $ and $ \gamma_i $ are non-negative and must not exceed the total available resources.

The formulated energy-aware resource management problem is challenging to solve due to the non-convexity and strong coupling among the decision variables. Thus, we first decouple the formulated problem of user $ i $ into the computing resource management problem and dataset offloading problem. Then, the GNEP is formulated for the uplink radio resource management at the edge server and the dataset offloading at the user $ i $ in which the coupling in their objective function and strategy sets are analyzed.

\section{Energy-aware Resource Management For MEC-enabled FL} \label{sec:solution}
In this section, the energy-aware resource management algorithm is presented where the model training and the resource management problem at the edge server and the mobile users are decoupled and solved alternatively. The uplink radio resource management problem of the edge server where the objective is to minimize the total time consumption while guaranteeing the energy limit of the mobile devices is defined as follows:
\begin{equation}\label{problem_edge_server_resource}
\begin{aligned}
& \underset{[\boldsymbol{\tilde{\omega}}, \boldsymbol{\bar{\omega}}] \in \mathcal{S}_E}{\text{minimize}}
& & t^{\text{total}} \\
& \text{subject to}
& & e_i^{\text{total}} \leq \Delta e_i, \forall i \in \mathcal{I}.
\end{aligned}
\end{equation} 
Thus, the edge server solves \eqref{equ:loss_function_edge} and \eqref{problem_edge_server_resource} alternatively. Moreover, the objective of the mobile user $ i $ is to minimize the total time consumption by optimizing the dataset offloading and computing resource management by ensuring its energy limit is as follows:

\begin{equation}\label{problem_local_resource}
\begin{aligned}
& \underset{[\delta_i, \gamma_i] \in \mathcal{S}_i}{\text{minimize}}
& & t^{\text{total}} \\
& \text{subject to}
& & e_i^{\text{total}} \leq \Delta e_i,
\end{aligned}
\end{equation}
where the mobile user $ i $ solves \eqref{equ:loss_function_mecFL} and \eqref{problem_local_resource} alternatively.
Due to the non-convexity and coupling among the dataset offloading $ \delta_i $, and the computing resource management $ \gamma_i $, the resource management problem of user $ i $ is decoupled into two independent problems, which are the dataset offloading problem and the computing resource management problem.

\subsection{Computing Resource Management for MEC-enabled FL}
Given the dataset offloading decision $ \delta_i $ and the uplink bandwidth allocation for the weight transmission $ \bar{\omega}_i $, the computing resource management problem of user $ i $ to minimize the total time consumption by taking its energy limit is defined as
\begin{equation}\label{problem_computing_resource}
\begin{aligned}
& \underset{\gamma_i \in \bar{\mathcal{S}}_i}{\text{minimize}}
& & t_i^{\text{local}} \\
& \text{subject to}
& & e_i^{\text{total}} \leq \Delta e_i,
\end{aligned}
\end{equation}
where $ \bar{\mathcal{S}}_i = \{ \gamma_i | 0 \leq \gamma_i \leq 1 \} $. As stated in Appendix \ref{proof_convexity}, the local time consumption $ t_i^{\text{local}} $ is decreasing in $ \gamma_i $, but the local energy consumption $ e_i^{\text{total}} $ is increasing with respect to $ \gamma_i $. Thus, the optimal value of $ \gamma_i $ exists when $ t_i^{\text{local}} $ can be decreased until $ e_i^{\text{total}} = \Delta e_i $. The closed form solution of $ \gamma_i $, which can be derived with the KKT conditions, is described as follows:
\begin{equation*}
\gamma_i = \left[ \frac{\Delta e_i - p_i \left( \frac{\delta_i f(|\mathcal{D}_i|)}{\tilde{\omega}_i R_i} + \frac{f(|\mathbf{w}_i|)}{\bar{\omega}_i R_i}\right)}{\psi (1 - \delta_i) f(|\mathcal{D}_i|) \tau (\Gamma_i)^2} \right]^{1/2},
\end{equation*}
\begin{equation}\label{computing_resource_sol}
\gamma_i^* = \max\{ \min \{ \gamma_i, 1 \}, 0 \},
\end{equation}
which is the projection of $ \gamma_i $ onto $ \bar{\mathcal{S}}_i $.
\begin{proof}
	Appendix \ref{proof_computing_resource_sol}.
\end{proof}

Thus, the optimal computing resource allocation of user $ i $ depends not only on its energy level but also the amount of datasets used for the local training and energy consumption for the uplink transmissions. The computing resource used for the local model training $ \gamma_i $ will be less if the energy level of the mobile user $ i $, $ \Delta e_i $ is low, or the total energy consumption of the uplink transmissions is high. Moreover, in order to guarantee the energy limit of the mobile device $ i $, less computing resource should be used for the model training if the amount of dataset used for the model training is large.

The energy limitation of the mobile users is assumed to be satisfied  by the computing resource management of the mobile users. Thus, we eliminate the energy limitation constraints of the mobile devices from the dataset offloading and uplink resource management problem. However, our proposed solution approach can satisfy the energy limitation of the mobile devices which is shown in the simulation results in Section \ref{sec:simulation_results}.

\subsection{Dataset Offloading Problem}
The dataset offloading of the mobile users and the uplink bandwidth management of the edge server can be formulated as follows. Given the computing resource allocation $ \gamma_i $, the uplink bandwidth management for the dataset offloading $ \tilde{\omega}_i $, and the weight transmission $ \bar{\omega}_i $, the dataset offloading decision of user $ i $ can be formulated as follows:
\begin{equation}\label{problem_offloading}
\underset{\delta_i \in \tilde{\mathcal{S}}_i}{\text{minimize}} \qquad t^{\text{total}},
\end{equation}
where $ \tilde{\mathcal{S}}_i = \{ \delta_i | 0 \leq \delta_i \leq 1 \} $. 
The dataset offloading problem \eqref{problem_offloading} can be rewritten as 
\begin{equation}\label{problem_offloading_couple}
\begin{aligned}
& \underset{\delta_i, \in \tilde{\mathcal{S}}_i}{\text{minimize}}
& & \Delta t \\
& \text{subject to}
& & t_i^{\text{local}} \leq \Delta t, \\
& & & t_i^{\text{edge}} \leq \Delta t,
\end{aligned}
\end{equation}
where $ t_i^{\text{edge}} = \frac{\delta_i f(|\mathcal{D}_i|)}{\tilde{\omega}_i R_i} + \frac{\sum_{i \in \mathcal{I}} \delta_i f(|\mathcal{D}_i|) \tau}{\Gamma_E} $.

\subsection{Uplink Resource Management Problem}
Given the dataset offloading $ \delta_i $ and the computing resource allocation $ \gamma_i, \forall i \in \mathcal{I} $, the uplink bandwidth allocation for the dataset offloading $ \boldsymbol{\tilde{\omega}} $ and the weight transmission $ \boldsymbol{\bar{\omega}} $ of the edge server can be formulated as follows:
\begin{equation}\label{problem_uplink}
\underset{[\boldsymbol{\tilde{\omega}}, \boldsymbol{\bar{\omega}}] \in \mathcal{S}_E}{\text{minimize}} \qquad t^{\text{total}}.
\end{equation}
 
 The uplink bandwidth management problem \eqref{problem_uplink} can be rewritten as
 \begin{equation}\label{problem_uplink_couple}
 \begin{aligned}
 & \underset{[\boldsymbol{\tilde{\omega}}, \boldsymbol{\bar{\omega}}] \in \mathcal{S}_E}{\text{minimize}}
 & & \Delta t \\
 & \text{subject to}
 & & t_i^{\text{local}} \leq \Delta t, \forall i \in \mathcal{I}, \\
 & & & t_i^{\text{edge}} \leq \Delta t, \forall i \in \mathcal{I}.
 \end{aligned}
 \end{equation} 
 
 As defined in \eqref{problem_offloading_couple} and \eqref{problem_uplink_couple}, the dataset offloading decision $ \delta_i $ is strongly dependent on the uplink bandwidth allocations $ \tilde{\omega}_i $ and $ \bar{\omega}_i $. Since the uplink bandwidth resource is limited, the management of the uplink radio resources among the mobile users with respect to their data offloading decisions is challenging. 
 
 \subsection{GNEP Formulation for Time Minimization}
 The GNEP formulation is defined in order to address the strong coupling among the uplink bandwidth and the dataset offloading, where not only the objective functions of the players but also their strategy sets are dependent on each other, as defined in \eqref{problem_offloading_couple} and \eqref{problem_uplink_couple}. Let $ \mathcal{P} = \{ 0, 1, 2, \cdots, I \} $ be the set of players in the time minimization game, where the edge server is indexed as $ 0 $, and the mobile user $ i $ as $ i = 1, 2, \cdots, I $. The action of the edge server, which is the uplink bandwidth management for the dataset offloading and the weight transmission, is denoted by $ \mathbf{x}_0 $, where $ \mathbf{x}_0 \coloneqq [ \boldsymbol{\tilde{\omega}}, \boldsymbol{\bar{\omega}} ] $. In addition, the action of user $ i, \forall i \in \mathcal{I} $, which is the dataset offloading decision, is denoted by $ \mathbf{x}_i $, where $ \mathbf{x}_i \coloneqq \delta_i $.
 
 The coupling constraints of the edge server and the mobile users are rewritten as follows. The local time consumption constraint is re-defined as
 \begin{equation}\label{equ:local_time_gnep}
 h_i^{\text{local}}(\delta_i, \gamma_i, \bar{\omega}_i) \leq 0,
 \end{equation}
 where $ h_i^{\text{local}}(\delta_i, \gamma_i, \bar{\omega}_i) = t_i^{\text{local}} - \Delta t $.
 The edge time consumption constraint is re-defined as
 \begin{equation}\label{equ:edge_time_gnep}
 h_i^{\text{edge}}(\boldsymbol{\delta}, \tilde{\omega}_i) \leq 0,
 \end{equation}
 where $ h_i^{\text{edge}}(\boldsymbol{\delta}, \tilde{\omega}_i) = t_i^{\text{edge}} - \Delta t $, $ \boldsymbol{\delta} = [\delta_i]_{i \in \mathcal{I}}^T $.
 
 The GNEP formulation of the edge server for the uplink bandwidth resource management is defined as
 \begin{equation}\label{gnep_uplink}
 \begin{aligned}
G_0(\mathbf{x}_{-0}) : \qquad & \underset{[\boldsymbol{\tilde{\omega}}, \boldsymbol{\bar{\omega}}] \in \check{\mathcal{S}}_0}{\text{minimize}}
 & & \Delta t \\
 & \text{subject to}
 & & h_i^{\text{local}}(\delta_i, \gamma_i, \bar{\omega}_i) \leq 0, \forall i \in \mathcal{I}, \\
 & & & h_i^{\text{edge}}(\boldsymbol{\delta}, \tilde{\omega}_i) \leq 0, \forall i \in \mathcal{I},
 \end{aligned}
 \end{equation} 
 where $ \check{\mathcal{S}}_0 = \mathcal{S}_E $. Let $ \hat{\mathcal{S}}_0 $ be the set of the coupling constraints of the edge server which is defined as 
 \begin{equation}\label{coupling_set_uplink}
 \hat{\mathcal{S}}_0 = \left \{ [\boldsymbol{\tilde{\omega}}, \boldsymbol{\bar{\omega}}] | h_i^{\text{local}}(\delta_i, \gamma_i, \bar{\omega}_i) \leq 0, h_i^{\text{edge}}(\boldsymbol{\delta}, \tilde{\omega}_i) \leq 0, \forall i \in \mathcal{I} \right \}.
 \end{equation}
 The compactness and convexity of $ G_0(\mathbf{x}_{-0}) $ is defined in the following lemma.
 \begin{lem}\label{g_0_convex}
 	$ \check{\mathcal{S}}_0 $ is compact and convex set. $ h_i^{\text{local}}(\delta_i, \gamma_i, \bar{\omega}_i) $ and $ h_i^{\text{edge}}(\boldsymbol{\delta}, \tilde{\omega}_i) $ are continous and convex in $ \boldsymbol{\tilde{\omega}} $, $ \boldsymbol{\bar{\omega}} $.
 \end{lem}
\begin{proof}
	Appendix \ref{proof_time}.
\end{proof}
 The GNEP formulation of the mobile user $ i $ for the dataset offloading is defined as
 \begin{equation}\label{gnep_offloading}
 \begin{aligned}
 G_i(\mathbf{x}_{-i}): \qquad & \underset{\delta_i, \in \check{\mathcal{S}}_i}{\text{minimize}}
 & & \Delta t \\
 & \text{subject to}
 & & h_i^{\text{local}}(\delta_i, \gamma_i, \bar{\omega}_i) \leq 0, \\
 & & & h_i^{\text{edge}}(\boldsymbol{\delta}, \tilde{\omega}_i) \leq 0,
 \end{aligned}
 \end{equation}
 where $ \check{\mathcal{S}}_i = \tilde{\mathcal{S}}_i $.
 Let $ \hat{\mathcal{S}}_i $ be the set of the coupling constraints of player $ i $ which is defined as
 \begin{equation}\label{coupling_set_offloading}
 \hat{\mathcal{S}}_i = \left \{ h_i^{\text{local}}(\delta_i, \gamma_i, \bar{\omega}_i) \leq 0, h_i^{\text{edge}}(\boldsymbol{\delta}, \tilde{\omega}_i) \leq 0 \right \}.
 \end{equation}
  The compactness and convexity of $ G_i(\mathbf{x}_{-i}) $ is defined in the following lemma.
 \begin{lem}\label{g_i_convex}
 	$ \check{\mathcal{S}}_i $ is compact and convex set. $ h_i^{\text{local}}(\delta_i, \gamma_i, \bar{\omega}_i) $ and $ h_i^{\text{edge}}(\boldsymbol{\delta}, \tilde{\omega}_i) $ are continuous and convex in $ \delta_i $.
 \end{lem}
 \begin{proof}
 	Appendix \ref{proof_time}.
 \end{proof}

\subsubsection{The Existence of the GNE} The GNE of the formulated generalized Nash game is defined as a point $ \mathbf{x}^* $ which solves $ G_p(\mathbf{x}_{-p}), \forall p \in \mathcal{P} $.
The existence of the GNE is stated by the following theorem.
\begin{thm}\label{gne_existence}
	There exists a generalized Nash equilibrium if the following conditions hold for $ p, \forall p \in \mathcal{P} $.
	\begin{itemize}
		\item $ \check{\mathcal{S}}_p, \forall p \in \mathcal{P}, $ are convex and compact sets.
		\item $ \hat{\mathcal{S}}_p, \forall p \in \mathcal{P}, $ is closed and convex.
		\item The objective function, which is $ \Delta t $, is continuous and convex with respect to $ \mathbf{x}_p $.
	\end{itemize}
\end{thm}

\subsubsection{Solution Approach to GNEP}
The GNE $ \mathbf{x}_p^* $ of the formulated GNEP $ G_p(\mathbf{x}_{-p}) $ can be derived by the KKT conditions. Thus, $ \mathbf{x}_0^* \coloneqq [\boldsymbol{\tilde{\omega}}, \boldsymbol{\bar{\omega}}] $ and $ \mathbf{x}_i^* \coloneqq \delta_i, \forall i \in \mathcal{I}, $ is as follows:
\begin{align}
\tilde{\omega}_i^* = & \frac{1}{\sum_{i \in \mathcal{I}} \left[ \frac{\tilde{\lambda}_i \delta_i f(|\mathcal{D}_i|)}{R_i} \right]^{1/2}} \left[ \frac{\tilde{\lambda}_i \delta_i f(|\mathcal{D}_i|)}{R_i} \right]^{1/2}, \label{uplink_data_sol} \\
\bar{\omega}_i^* = & \frac{1}{\sum_{i \in \mathcal{I}} \left[ \frac{\bar{\lambda}_i f(|\mathbf{x}_i|)}{R_i} \right]^{1/2}} \left[ \frac{\bar{\lambda}_i f(|\mathbf{w}_i|)}{R_i} \right]^{1/2}, \label{uplink_weight_sol} \\
\delta_i^* = & \left[ \frac{f(|\mathcal{D}_i|)}{\tilde{\omega}_i R_i} + \frac{f(|\mathcal{D}_i \tau|)}{\Gamma_E} + \frac{f(|\mathcal{D}_i|) \tau}{\gamma_i \Gamma_i} \right]^{-1} \nonumber \\ & \left[ \frac{f(|\mathcal{D}_i|) \tau}{\gamma_i \Gamma_i} + \frac{f(|\mathbf{w}_i|)}{\bar{\omega}_i R_i} - \frac{\sum_{j \in \mathcal{I}, j \neq i} \delta_j f(|\mathcal{D}_j|) \tau}{\Gamma_E} \right], \label{offloading_sol}
\end{align}
where $ \tilde{\lambda}_i $ and $ \bar{\lambda}_i $ are the Lagrange multipliers associated to two coupling constraints for user $ i $.

\begin{proof}
	Appendix \ref{proof_game_sol}.
\end{proof}

The Lagrange multipliers, $ \tilde{\lambda}_i $ and $ \bar{\lambda}_i, \forall i \in \mathcal{I} $, are updated as follows.
\begin{align}
\tilde{\lambda}_i^{k} &= \left\{\begin{matrix}
\tilde{\lambda}_i^{k-1} + \Delta_i &  \text{if } e_i > \Delta e_i, \\ 
\tilde{\lambda}_i^{k-1} & \text{if } e_i \leq \Delta e_i,
\end{matrix}\right. \label{equ:edge_lambda_update} \\
\bar{\lambda}_i^{k} &= 1 - \tilde{\lambda}_i^{k}, \label{equ:locallambda_update}
\end{align}
where $ \Delta_i $ is an increment parameter of user $ i $.
As defined in \eqref{uplink_data_sol} and \eqref{uplink_weight_sol}, the Lagrange multipliers $ \tilde{\lambda}_i $ and $ \bar{\lambda}_i $ act as a weight parameter to the proportional resource allocation of the uplink bandwidth. Since $ e_i $ is decreasing in $ \tilde{\omega}_i $ and $ \bar{\omega}_i $, the energy limitation of user $ i $ can be satisfied by allocating more uplink bandwidth to user $ i $ in the dataset offloading which has the higher energy consumption than the weight transmission.

\subsection{Energy-aware Resource Management Algorithm}
The energy-aware resource management algorithm is proposed for the joint learning, dataset offloading, computing, and uplink resource management for the MEC-enabled FL which works as follows. First, the initial value for the Lagrange multipliers, dataset offloading, computing and uplink resource allocation are chosen. Users offload a portion of their local dataset to the edge server as stated at line \ref{alg_line:offload_initial}. Users and the edge server perform the model training simultaneously as defined at line \ref{alg_line:train_initial}. Once the weight update of the all users is received at the edge server, the model aggregation is performed as stated at line \ref{alg_line:aggregate_initial}. The Lagrange multipliers, the dataset offloading, computing and uplink resource allocation are then updated. In addition, the edge server and the users perform the model training and the updated weights are aggregated. This process is repeated until convergence as defined in lines \ref{alg_line:loop_beginning}-\ref{alg_line:loop_end}. Since all mobile users and the edge server implement the best response strategy, the proposed algorithm will converge to a stationary point.
\begin{algorithm}[t]
	\caption{Energy-aware Resource Management Algorithm}\label{alg:resource_management}
	\floatname{algorithm}{Procedure}
	\begin{algorithmic}[1]
		\State Choose an initial value for the Lagrange multipliers $ \tilde{\lambda}_i^0, \bar{\lambda}_i^0, \forall i \in \mathcal{I} $.
		\State $ k \gets 0 $. 
		\State Choose an initial value for $ \delta_i^k, \gamma_i^k, \tilde{\omega}_i^k, \bar{\omega}_i^k, \forall i \in \mathcal{I} $. \label{alg:initial}
		\State User $ i, \forall i \in \mathcal{I} $ offload $ \delta_i^k $ of its dataset to the edge server. \label{alg_line:offload_initial}
		\State User $ i, \forall i \in \mathcal{I}, $ and the edge server perform the model training simultaneously. \label{alg_line:train_initial}
		\State User $ i, \forall i \in \mathcal{I} $ uploads its weight parameters.
		\State The model aggregation is performed at the edge server. \label{alg_line:aggregate_initial}
		\Repeat
			\State $ k \gets k + 1 $. \label{repeat} \label{alg_line:loop_beginning}
			\State At user $ i, \forall i \in \mathcal{I} $,
			\State $ \gamma_i^k \gets \gamma_i^* $ as defined in \eqref{computing_resource_sol}.
			\State $ \delta_i^k \gets \delta_i^* $ as defined in \eqref{offloading_sol}.
			\State $ \delta_i^k $ of the dataset is offloaded to the edge server regarding $ \tilde{\omega}_i^{k-1} $. 
			\State The model training is performed regarding $ \delta_i^k $ and $ \gamma_i^k $.
			\State $ \mathbf{w}_i^k $ is uploaded to the edge server regarding $ \bar{\omega}_i^k $.
			\State At the edge server,
			\State $ \tilde{\lambda}_i^k, \bar{\lambda}_i^k, \forall i \in \mathcal{I}, $ are updated as defined in \eqref{equ:edge_lambda_update} and \eqref{equ:locallambda_update}.
			\State $ \tilde{\omega}_i^k \gets \tilde{\omega}_i^*, \forall i \in \mathcal{I}, $ as defined in \eqref{uplink_data_sol}.
			\State $ \bar{\omega}_i^k \gets \bar{\omega}_i^*, \forall i \in \mathcal{I}, $ as defined in \eqref{uplink_weight_sol}.
			\State The model training is performed regarding $ \boldsymbol{\delta}^k $.
			\State The model aggregation is performed. \label{alg_line:loop_end}
		\Until{$ \abs{\hat{l}^k - \hat{l}^{k-1}}\leq \epsilon $ and $ \abs{ t^{\text{total}, k} - t^{\text{total}, k - 1} } \leq \epsilon $.}
	\end{algorithmic}
\end{algorithm}
The loss value $ \hat{l} $ is defined as $ \hat{l} = \sum_{j \in \hat{\mathcal{D}}} l(\bar{\mathbf{w}}, x_j, y_j) $, where it calculates the testing loss of the final model on the test dataset $ \hat{\mathcal{D}} $.

\section{Simulation Results}\label{sec:simulation_results}
We consider a single-cell macro base station deployed together with an edge server for the model training and aggregation. MINIST dataset is used for the model training where the logistic regression is performed on 50 mobile users. The data samples in the whole dataset are randomly shuffled and distributed among users where each user has approximately 1,200 data samples. The total available uplink bandwidth is considered as 20 MHz. The edge server is equipped with 16 GHz CPU. We consider the system heterogeneity in the simulation where the mobile users have different CPU frequency and energy limitation where the CPU frequency of the mobile users follows a uniform distribution of [1.2, 1.5] GHz. The energy limit of the mobile users is considered to follow a uniform distribution as well which is [45, 60] watt. We compare the traditional FL and proposed MEC-enabled FL where the dataset offloading for the edge training is not allowed in the traditional FL which uses all the data samples in the local dataset for the training. The loss value in the figures is the testing loss on the final model. For the proposed energy-aware resource management algorithm, the initial points for the dataset offloading and computing resource allocation are chosen randomly while the uniform allocation is performed on the uplink bandwidth resources. To compare the traditional and proposed MEC-enabled FL in terms of the size of offloaded dataset and computing resource, the uniform allocation is used for the uplink bandwidth resource management.

\begin{figure}[t!] 
	\centering
	\includegraphics[keepaspectratio=true,scale=0.26]{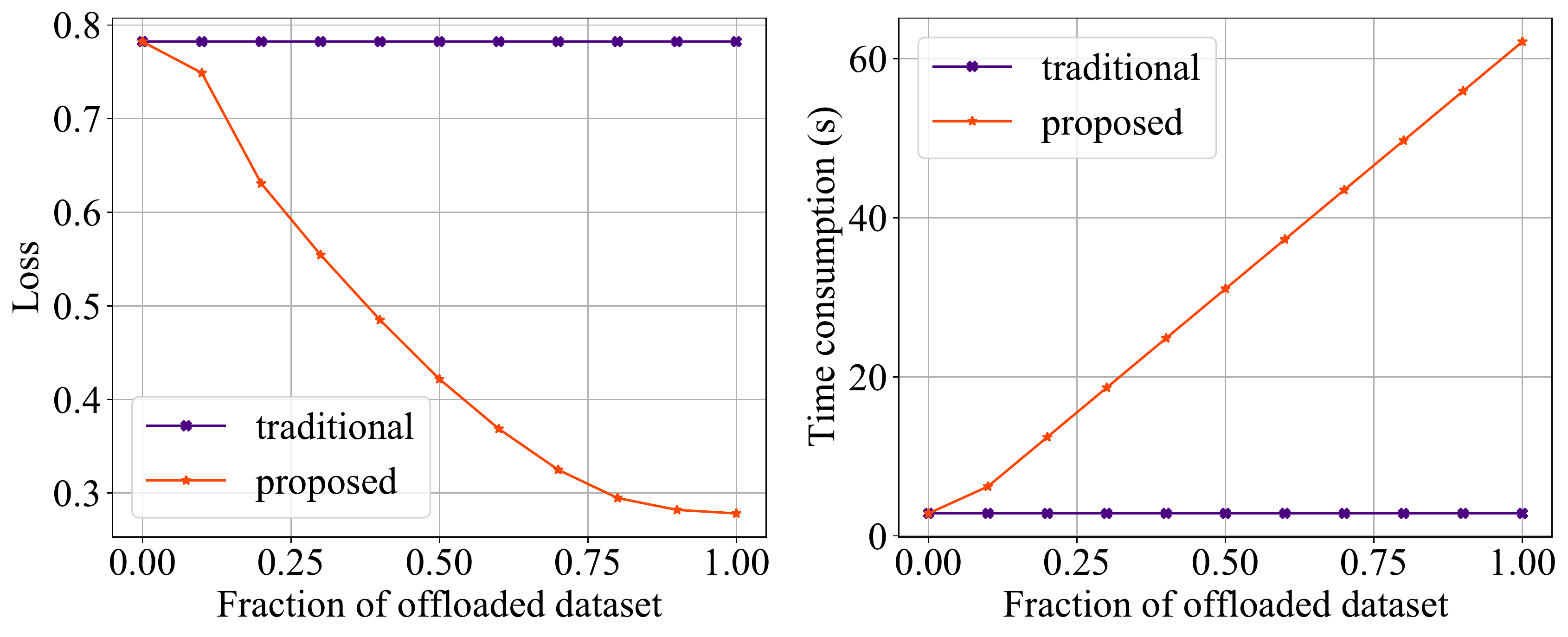}
	\caption{Loss and time consumption with respect to the fraction of offloaded data.}
	\label{fig:loss_time_proposed_tradition}
\end{figure}

Fig. \ref{fig:loss_time_proposed_tradition} shows the comparison of the traditional and proposed FL on the size of offloaded dataset regarding loss and total time consumption. Since the dataset offloading is not allowed in the traditional FL, the loss and total time consumption are constant across the size of the offloaded dataset. The proposed MEC-enabled FL can achieve the better model performance by offloading the portion of local dataset where the edge server performs the model training on all the local dataset offloaded from the user. The proposed MEC-enabled FL is same as the traditional FL when the fraction of the offloaded data is zero which means the mobile users use all the dataset for the local training. When the fraction is 1, the proposed FL is same as the centralized model training where all the local datasets are used for the model training at edge server. Since the users need to upload the local datasets to the edge server, the total time consumption of the proposed model is higher than the traditional FL but it can be minimized by the decent resource management approach.

\begin{figure}[t!]
	\centering
	\includegraphics[keepaspectratio=true,scale=0.228]{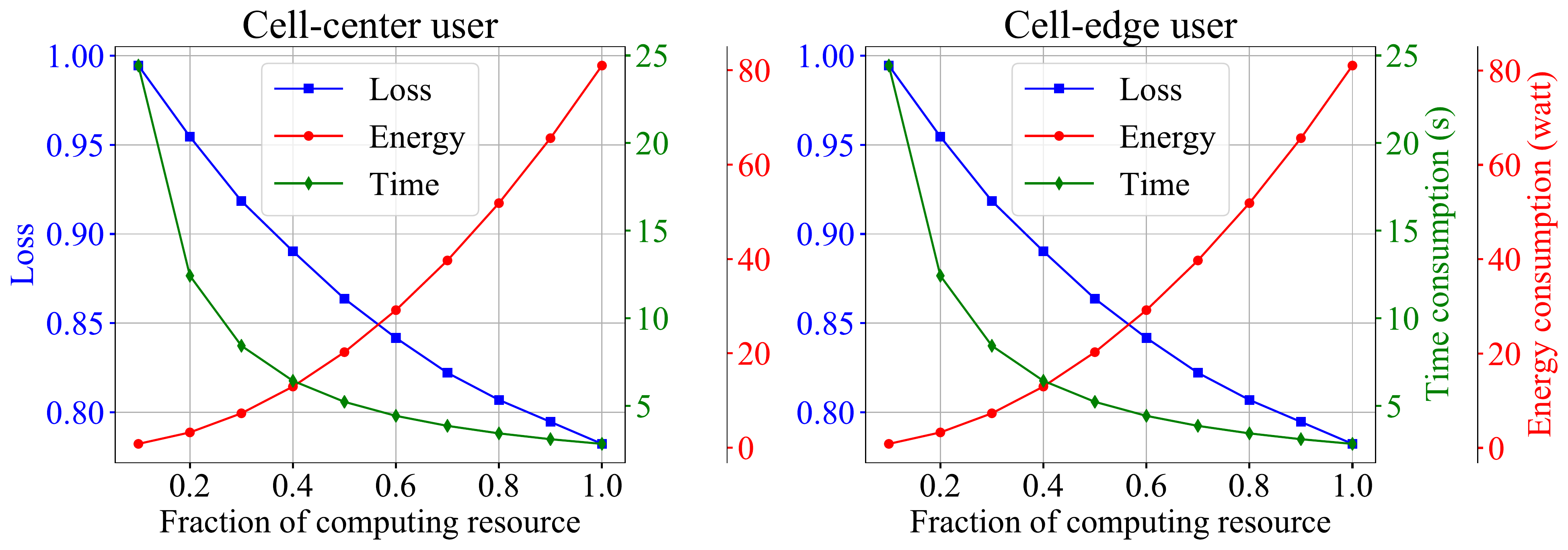}
	\caption{Loss, energy, and time consumption with respect to the fraction of computing resource allocated.}
	\label{fig:loss_energy_time_wrt_computing_resource}
\end{figure}

\begin{figure}[t!]
	\centering
	\includegraphics[keepaspectratio=true,scale=0.26]{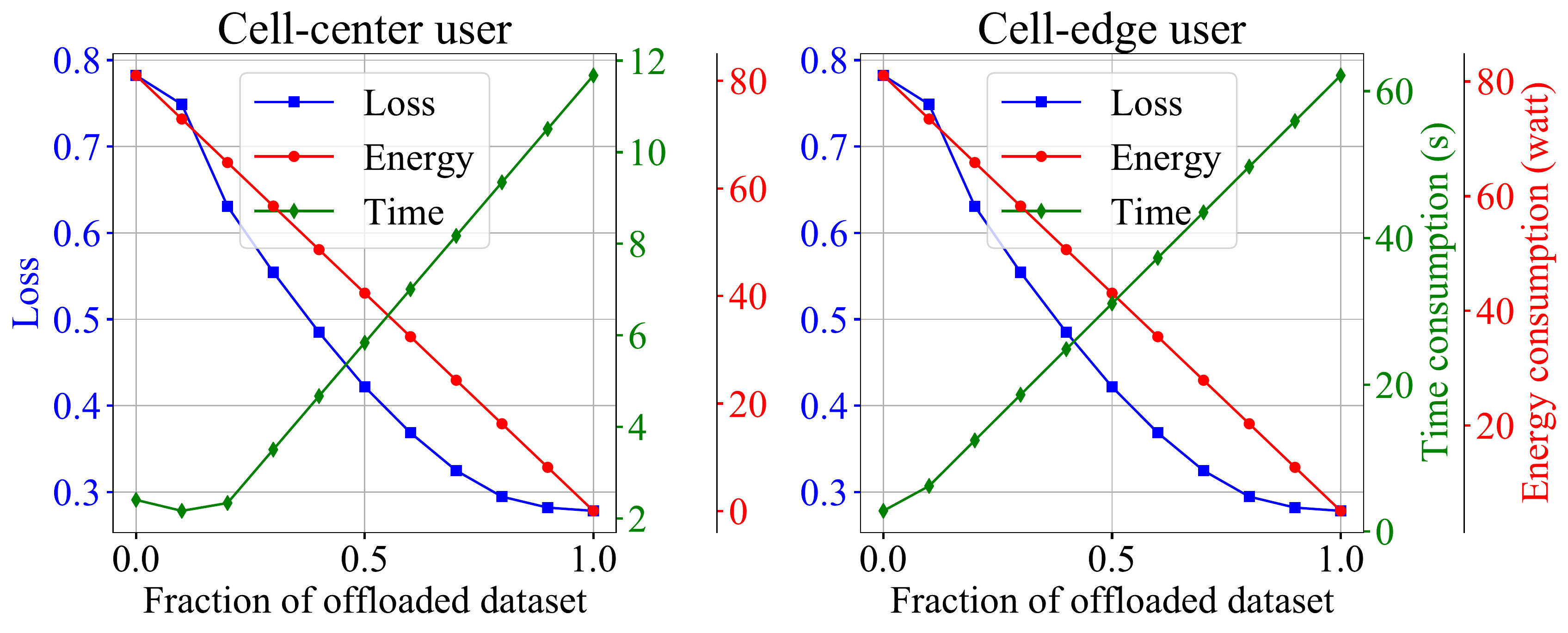}
	\caption{Loss, energy, and time consumption with respect to the fraction of offloaded data.}
	\label{fig:loss_energy_time_wrt_data_offload}
\end{figure}

Fig. \ref{fig:loss_energy_time_wrt_computing_resource} and Fig. \ref{fig:loss_energy_time_wrt_data_offload} show the testing loss, energy and time consumption of the proposed model with respect to the computing resources and dataset offloading for the cell-center and cell-edge users. As more computing resources are used for the local model training, the loss and time consumption is decreased. However, the energy consumption of the mobile device will increase with the computing resource allocation. As for the amount of offloaded dataset, the loss and energy decreases as the size of offloaded dataset increases since the energy consumption of the model training is much higher than that of transmissions. But, the total time consumption is increased with respect to the offloaded dataset due to the time taken for dataset offloading. The total time taken for the cell-center user is much lower than that of the cell-edge user. Since the synchronous update is used for the model aggregation, the cell-edge user has the high impact on the total time consumption.

\begin{figure}[t!]
	\centering
	\includegraphics[keepaspectratio=true,scale=0.26]{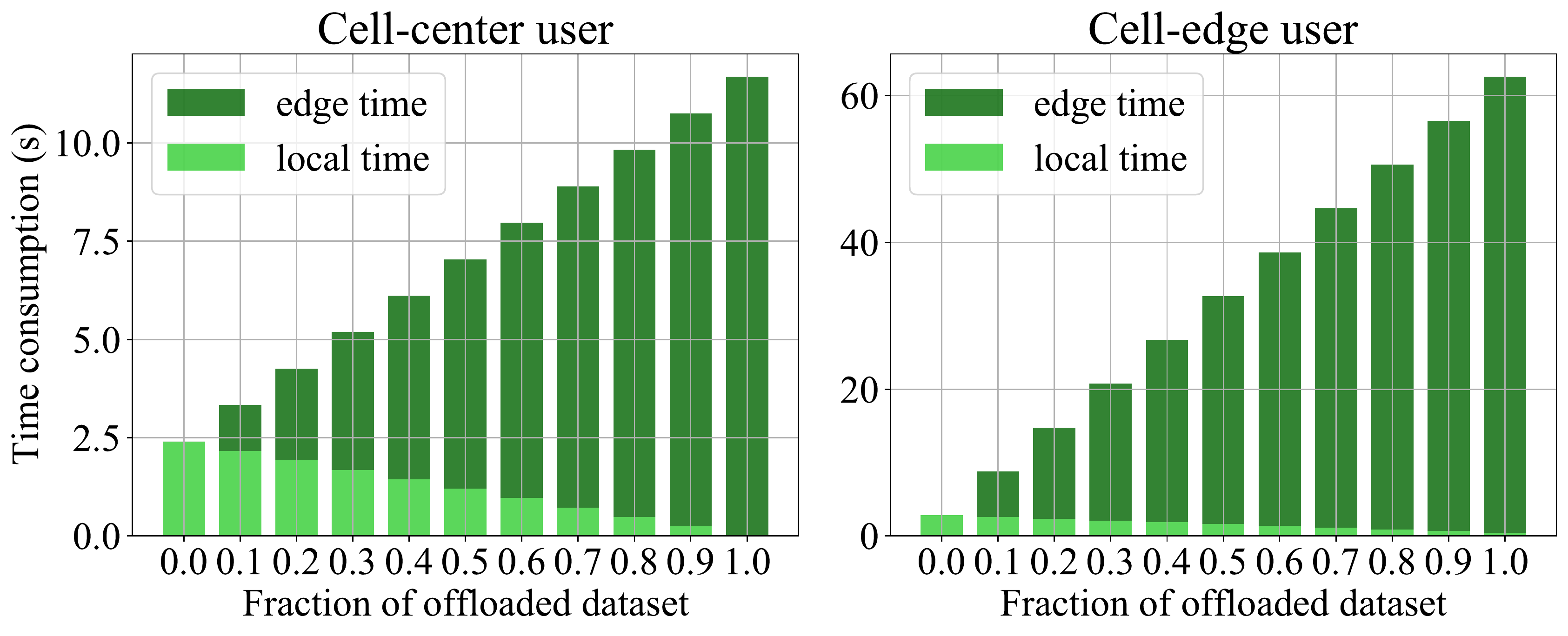}
	\caption{Comparison of local time taken for (i) local training, (ii) weight transmission, and edge time taken for (i) dataset offloading, (ii) edge training, with respect to the fraction of offloaded data.}
	\label{fig:time_comparison_wrt_offloaded_data}
\end{figure}

Fig. \ref{fig:time_comparison_wrt_offloaded_data} shows the time taken for the local and edge training for the cell-center and cell edge user. The time taken for the local training gets lower than that for the edge training as the offloaded data size increases. Due to the poor channel condition, the cell-edge user needs higher time consumption in the dataset offloading than the cell-center user. Thus, the decent resource management approach is required to minimize the total time consumption.

\begin{figure}[t!]
	\centering
	\includegraphics[keepaspectratio=true,scale=0.26]{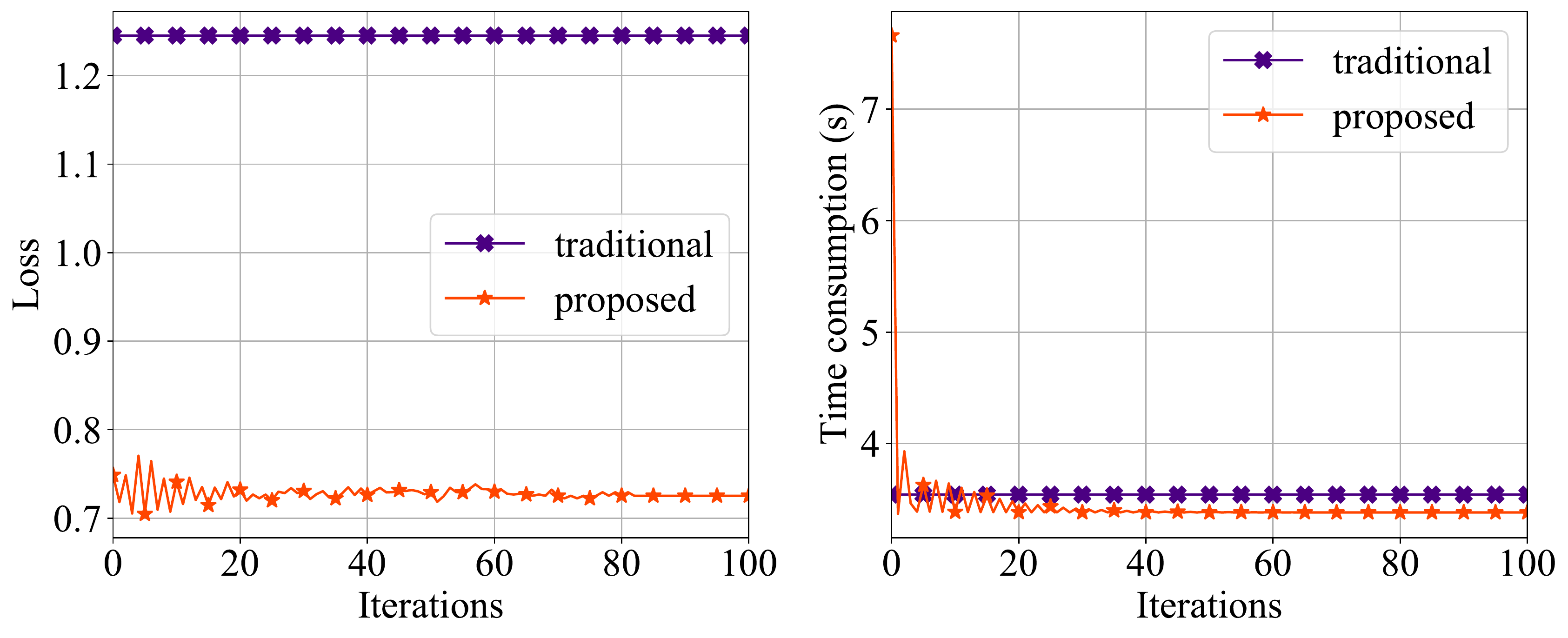}
	\caption{Comparison of traditional and proposed MEC-enabled FL in loss and time consumption.}
	\label{fig:loss_time_iterations}
\end{figure}

Fig. \ref{fig:loss_time_iterations} shows the comparison of the traditional and proposed FL on Algorithm \ref{alg:resource_management} where the algorithm converges to a stationary point after a few iterations. The proposed MEC-enabled FL performs better than the traditional FL since the offloaded local datasets are trained collectively at the edge server. With the proposed resource management approach, the MEC-enabled FL can achieve lower time consumption than the traditional FL. Fig. \ref{fig:loss_energy_time_convergence} shows the convergence of the algorithm in the final model loss, total time consumption and individual energy consumption of the cell-center and cell-edge user. The energy consumption of the cell-center user fluctuates more than that of cell-edge user where the energy limit of the mobile users are guaranteed eventually.

\begin{figure}[t!]
	\centering
	\includegraphics[keepaspectratio=true,scale=0.22]{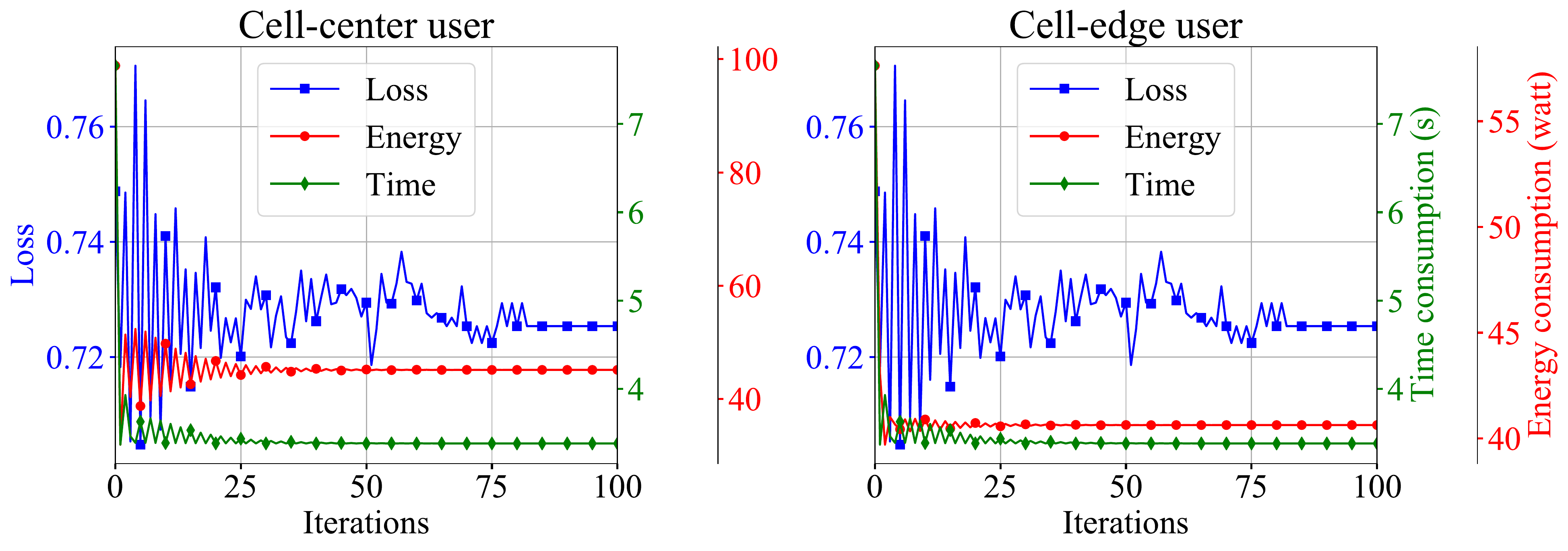}
	\caption{Convergence of the algorithm in loss, energy and time consumption.}
	\label{fig:loss_energy_time_convergence}
\end{figure}

\begin{figure}[t!] 
	\centering
	\includegraphics[keepaspectratio=true,scale=0.26]{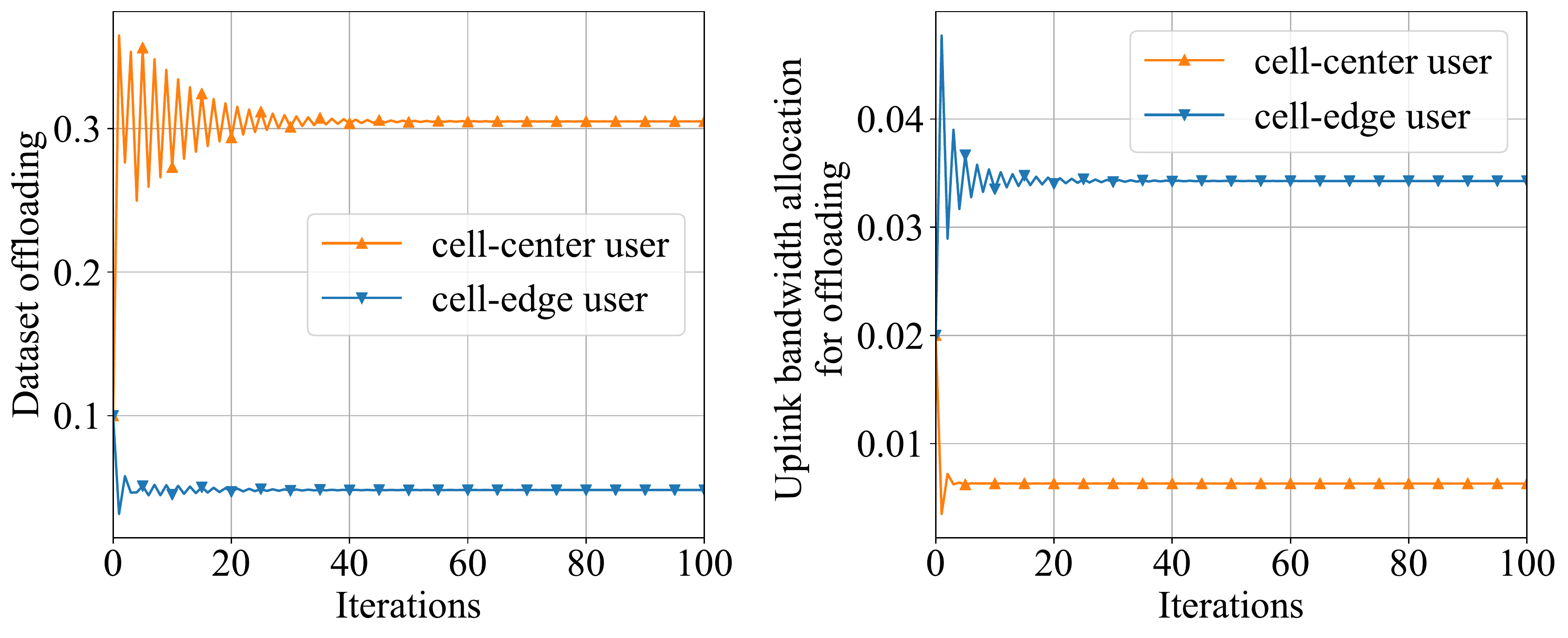}
	\caption{Convergence of the algorithm in data offloading and uplink bandwidth allocation for offloading.}
	\label{fig:data_offloading_uplink_convergence}
\end{figure}

\begin{figure}[t!] 
	\centering
	\includegraphics[keepaspectratio=true,scale=0.26]{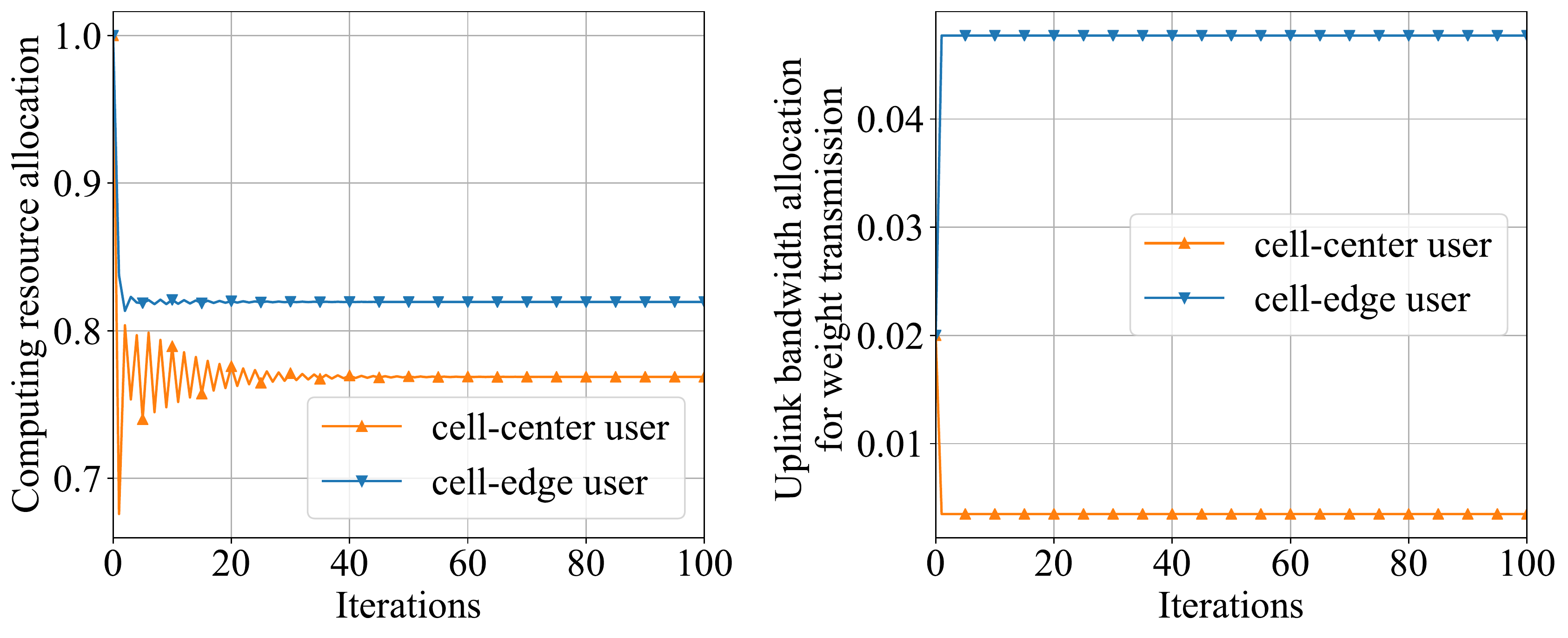}
	\caption{Convergence of the algorithm in computing resource and uplink bandwidth allocation for weight transmission.}
	\label{fig:computing_resource_uplink_convergence}
\end{figure}

Fig. \ref{fig:data_offloading_uplink_convergence} and Fig. \ref{fig:computing_resource_uplink_convergence} shows the comparison of the resource management of the Algorithm \ref{alg:resource_management} at the cell-center and cell-edge users. Due to the poor channel condition of the cell-edge user, less local dataset is offloaded to the edge server while more uplink resource is required for the dataset offloading than the cell-center user in order to minimize the total time consumption. Similar to the uplink resource allocation for the dataset offloading, the cell-edge user require more uplink bandwidth for the weight transmission than the cell-center user to reduce the total time consumption. Since the cell-edge user offload a small portion of its local datasets to the edge server, the higher amount of computing resource is required for the local training than that of the cell-center user which is shown in Fig. \ref{fig:computing_resource_uplink_convergence}.

\section{Conclusion}\label{sec:conclusion}
In this paper, the energy-aware resource management for the MEC-enabled FL is proposed. In particular, the mobile users are allowed to offload a portion of their local dataset to the MEC server; and hence, the tradeoff between the performance of the training model and the energy consumption at user devices with respect to the amount of data samples used for the local training is handled. To that end, an energy-aware  resource  management  problem  is formulated with the objective of minimizing the training loss and time consumption, while satisfying the device's energy constraints. The formulated problem is decoupled into multiple sub-problems due to the coupling between the decision variables. Then, the solution for the computing resource management is derived by ensuring the energy budget of the mobile users. Moreover, the problem of dataset offloading and uplink resource management is formulated as a GNEP, and the existence of a GNE is derived. The solution to the dataset offloading and uplink bandwidth allocation is derived to minimize the total time consumption. To that end, the energy-aware resource management algorithm is proposed. Finally, extensive simulations are performed which show that the total time consumption of the proposed MEC-enabled FL model is competitively lower than the traditional FL approach when adopting the proposed resource management algorithm.


%
\appendices
\section{Proof of Convexity}\label{proof_convexity}
\subsection{Energy Consumption of users}\label{proof_energy}
The first derivatives of the total energy consumption of user $ i $ with respect to $ \delta_i $, $ \gamma_i $, $ \tilde{\omega}_i $, and $ \bar{\omega}_i $ are as follows:
\begin{equation*}
\begin{aligned}
\frac{\partial e_i}{\partial \delta_i} &= \frac{p_i f(|\mathcal{D}_i|)}{\tilde{\omega}_i R_i} - \psi f(|\mathcal{D}_i|) \tau \left( \gamma_i \Gamma_i \right)^2, \\
\frac{\partial e_i}{\partial \gamma_i} &= 2 \psi (1 - \delta_i) f(|\mathcal{D}_i|) \tau \gamma_i \Gamma_i^2, \\
\frac{\partial e_i}{\partial \tilde{\omega}_i} &= - \frac{p_i \delta_i f(|\mathcal{D}_i|)}{\tilde{\omega}_i^2 R_i}, \\
\frac{\partial e_i}{\partial \bar{\omega}_i} &= - \frac{p_i f(|\mathbf{w}_i|)}{\bar{\omega}_i^2 R_i}.
\end{aligned}
\end{equation*}

The second derivatives of $ e_i $ with respect to $ \delta_i $, $ \gamma_i $, $ \tilde{\omega}_i $, and $ \bar{\omega}_i $ are as follows:
\begin{equation*}
\begin{aligned}
\frac{\partial^2 e_i}{\partial \delta_i^2} &= 0, \\
\frac{\partial^2 e_i}{\partial \gamma_i^2} &= 2 \psi (1 - \delta_i) f(|\mathcal{D}_i|) \tau \Gamma_i^2 \geq 0, \\
\frac{\partial^2 e_i}{\partial \tilde{\omega}_i^2} &= \frac{2 p_i \delta_i f(|\mathcal{D}_i|)}{\tilde{\omega}_i^3 R_i} \geq 0, \\
\frac{\partial^2 e_i}{\partial \bar{\omega}_i^2} &= \frac{2 p_i f(|\mathbf{w}_i|)}{\bar{\omega}_i^3 R_i} \geq 0.
\end{aligned}
\end{equation*}
The second derivatives of $ e_i $ are non-negative which are $ \partial^2 e_i \geq 0 $ because $ 0 \leq \delta_i \leq 1 $, $ 0 \leq \tilde{\omega}_i \leq 1 $ and $ 0 \leq \bar{\omega}_i \leq 1 $.
Thus, $ e_i $ is convex with respect to $ \delta_i $, $ \gamma_i $, $ \tilde{\omega}_i $ and $ \bar{\omega}_i $.

\subsection{Time Consumption}\label{proof_time}
The time consumption for the edge model training of user $ i $ includes the time taken for the datasets offloading and model traing at the edge server which is defined as
\begin{equation*}
t_i^{\text{edge}} = \frac{\delta_i f(|\mathcal{D}_i|)}{\tilde{\omega}_i R_i} + \frac{\sum_{i \in \mathcal{I}} \delta_i f(|\mathcal{D}_i|) \tau}{\Gamma_E}.
\end{equation*}
The time taken for the edge training $ t_i^{\text{edge}} $ is affected by only the dataset offloading $ \delta_i $ and $ \tilde{\omega}_i $. Thus, the first derivatives of $ t_i^{\text{edge}} $ with respect to $ \delta_i $ and $ \tilde{\omega}_i $ are
\begin{equation*}
\begin{aligned}
\frac{\partial t_i^{\text{edge}}}{\partial \delta_i} &= \frac{f(|\mathcal{D}_i|)}{\tilde{\omega}_i R_i} + \frac{f(|\mathcal{D}_i|) \tau}{\Gamma_E}, \\
\frac{\partial t_i^{\text{edge}}}{\partial \tilde{\omega}_i} &= - \frac{\delta_i f(|\mathcal{D}_i|)}{\tilde{\omega}_i^2 R_i}.
\end{aligned}
\end{equation*}
The second derivatives of $ t_i^{\text{edge}} $ with respect to $ \delta_i $ and $ \tilde{\omega}_i $ are
\begin{equation*}
\begin{aligned}
\frac{\partial^2 t_i^{\text{edge}}}{\partial \delta_i^2} &= 0, \\
\frac{\partial^2 t_i^{\text{edge}}}{\partial \tilde{\omega}_i^2} &= \frac{2 \delta_i f(|\mathcal{D}_i|)}{\tilde{\omega}_i^3 R_i} \geq 0.
\end{aligned}
\end{equation*}
Thus, $ t_i^{\text{edge}} $ is convex with respect to $ \delta_i $ and $ \tilde{\omega}_i $.

The time taken for the local model training $ t_i^{\text{local}} $ is affected by the dataset offloading $ \delta_i $, the computing resource allocation $ \gamma_i $, and the uplink bandwidth allocation for the weight transmission $ \bar{\omega}_i $. Thus, the first derivatives of $ t_i^{\text{local}} $ with respect to $ \delta_i, \gamma_i $, and $ \bar{\omega}_i $ are
\begin{equation*}
\begin{aligned}
\frac{\partial t_i^{\text{local}}}{\partial \delta_i} &= - \frac{f(|\mathcal{D}_i|) \tau}{\gamma_i \Gamma_i}, \\
\frac{\partial t_i^{\text{local}}}{\partial \gamma_i} &= - \frac{(1 - \delta_i) f(|\mathcal{D}_i|) \tau}{\gamma_i^2 \Gamma_i}, \\
\frac{\partial t_i^{\text{local}}}{\partial \bar{\omega}_i} &= - \frac{f(|\mathbf{w}_i|)}{\bar{\omega}_i^2 R_i}.
\end{aligned}
\end{equation*}
The second derivatives of $ t_i^{\text{local}} $ with respect to $ \delta_i, \gamma_i $, and $ \bar{\omega}_i $ are
\begin{equation*}
\begin{aligned}
\frac{\partial^2 t_i^{\text{local}}}{\partial \delta_i^2} &= 0, \\
\frac{\partial^2 t_i^{\text{local}}}{\partial \gamma_i^2} &= \frac{2 (1 - \delta_i) f(|\mathcal{D}_i|) \tau}{\gamma_i^3 \Gamma_i} \geq 0, (\because 0 \leq \gamma_i \leq 1), \\
\frac{\partial^2 t_i^{\text{local}}}{\partial \bar{\omega}_i^2} &= \frac{2 f(|\mathbf{w}_i|)}{\bar{\omega}_i^3 R_i} \geq 0.
\end{aligned}
\end{equation*}
Thus, $ t_i^{\text{local}} $ is convex in $ \delta_i, \gamma_i $ and $ \bar{\omega}_i $.

\section{Derivation of Optimal Solutions}\label{optimal_solutions}
\subsection{Optimal Computing Resource Derivation}\label{proof_computing_resource_sol}
To derive the closed form solution of the computing resource management $ \gamma_i $ of user $ i $, the Lagrangian of \eqref{problem_computing_resource} is defined as follows:
\begin{equation*}
\begin{aligned}
\mathcal{L}(\gamma_i) =& \frac{(1 - \delta_i) f(|\mathcal{D}_i|) \tau}{\gamma_i \Gamma_i} + \frac{f(|\mathbf{w}_i|)}{\bar{\omega}_i R_i}  \\ &+ \beta \left[ \psi (1 - \delta_i) f(|\mathcal{D}_i|) \tau (\gamma_i \Gamma_i)^2 - \Delta e_i \right].
\end{aligned}
\end{equation*}
The first derivative of $ \mathcal{L}(\gamma_i) $ with respect to $ \gamma_i $ is
\begin{equation*}
\frac{\partial \mathcal{L}(\gamma_i)}{\partial \gamma_i} = - \frac{(1 - \delta_i) f(|\mathcal{D}_i|) \tau}{\gamma_i^2 \Gamma_i} + 2 \beta \psi (1 - \delta_i) f(|\mathcal{D}_i|) \tau \gamma_i \Gamma_i^2.
\end{equation*}
By setting the first derivative of $ \mathcal{L}(\gamma_i) $ to zero, which is $ \frac{\partial \mathcal{L}(\gamma_i)}{\partial \gamma_i} = 0 $, 
\begin{equation*}
\gamma_i = \left[ \frac{(1- \delta_i) f(|\mathcal{D}_i|)}{2 \beta \psi (1 - \delta_i) f(|\mathcal{D}_i|) \Gamma_i^3} \right]^{1/3}.
\end{equation*}
If $ \beta = 0 $, $ \gamma_i $ would be undefined. Thus, $ \beta $ must be greater than zero which makes the following condition: $ e_i^{\text{total}} = \Delta e_i $ according to the KKT conditions. Thus, the computing resource management of user $ i $ can be derived from $ e_i^{\text{total}} = \Delta e_i $ as follows:
\begin{equation*}
\gamma_i = \left[ \frac{\Delta e_i - p_i \left( \frac{\delta_i f(|\mathcal{D}_i|)}{\tilde{\omega}_i R_i} + \frac{f(|\mathbf{w}_i|)}{\bar{\omega}_i R_i}\right)}{\psi (1 - \delta_i) f(|\mathcal{D}_i|) \tau (\Gamma_i)^2} \right]^{1/2}.
\end{equation*}

\subsection{Optimal Dataset Offloading and Uplink Bandwidth Derivation}\label{proof_game_sol}
The Lagrangian of \eqref{gnep_offloading} is defined as follows.
\begin{equation*}
\begin{aligned}
\mathcal{L}(\delta_i) =& \Delta t + \tilde{\lambda}_i \left[ \frac{\delta_i f(|\mathcal{D}_i|)}{\tilde{\omega}_i R_i} + \frac{\sum_{i \in \mathcal{I}} \delta_i f(|\mathcal{D}_i|) \tau}{\Gamma_E} - \Delta t \right] \\ &+ \bar{\lambda}_i \left[ \frac{(1 - \delta_i) f(|\mathcal{D}_i|) \tau}{\gamma_i \Gamma_i} + \frac{f(|\mathbf{w}_i|)}{\bar{\omega}_i R_i} - \Delta t \right].
\end{aligned}
\end{equation*}
The first derivatives of $ \mathcal{L}(\delta_i) $ with respect to $ \delta_i $ is as follows.
\begin{equation*}
\frac{\partial \mathcal{L}(\delta_i)}{\partial \delta_i} = \frac{\tilde{\lambda}_i f(|\mathcal{D}_i|)}{\tilde{\omega}_i R_i} + \frac{\tilde{\lambda}_i f(|\mathcal{D}_i|) \tau}{\Gamma_E} - \frac{\bar{\lambda}_i f(|\mathcal{D}_i|) \tau}{\gamma_i \Gamma_i},
\end{equation*}
where the Lagrangian, $ \mathcal{L}(\delta_i) $ is linear in $ \delta_i $.

To derive the solution of $ \tilde{\omega}_i, \bar{\omega}_i $, the Langrangian of \eqref{gnep_uplink} is defined as follows.
\begin{equation*}
\begin{aligned}
\mathcal{L}(\boldsymbol{\tilde{\omega}}, \boldsymbol{\bar{\omega}}) =& \Delta t + \sum_{i \in \mathcal{I}} \tilde{\lambda}_i \left[ \frac{\delta_i f(|\mathcal{D}_i|)}{\tilde{\omega}_i R_i} + \frac{\sum_{i \in \mathcal{I}} \delta_i f(|\mathcal{D}_i|) \tau}{\Gamma_E} - \Delta t \right] \\ &+ \sum_{i \in \mathcal{I}} \bar{\lambda}_i \left[ \frac{(1 - \delta_i) f(|\mathcal{D}_i|) \tau}{\gamma_i \Gamma_i} + \frac{f(|\mathbf{w}_i|)}{\bar{\omega}_i R_i} - \Delta t \right] \\ & + \tilde{\mu} \left[ \sum_{i \in \mathcal{I}} \tilde{\omega}_i - 1 \right] + \bar{\mu} \left[ \sum_{i \in \mathcal{I}} \bar{\omega}_i - 1 \right].
\end{aligned}
\end{equation*}
The first derivatives of $ \mathcal{L}(\boldsymbol{\tilde{\omega}}, \boldsymbol{\bar{\omega}}) $ with respect to $ \tilde{\omega}, \bar{\omega} $ are
\begin{align*}
\frac{\partial \mathcal{L}(\boldsymbol{\tilde{\omega}}, \boldsymbol{\bar{\omega}})}{\partial \tilde{\omega}_i} &= - \frac{\tilde{\lambda}_i \delta_i f(|\mathcal{D}_i|)}{\tilde{\omega}_i^2 R_i} + \tilde{\mu}, \\
\frac{\partial \mathcal{L}(\boldsymbol{\tilde{\omega}}, \boldsymbol{\bar{\omega}})}{\partial \bar{\omega}_i} &= - \frac{\bar{\lambda}_i f(|\mathbf{w}_i|)}{\bar{\omega}_i^2 R_i} + \bar{\mu}.
\end{align*}
By setting the first derivatives of $ \mathcal{L}(\boldsymbol{\tilde{\omega}}, \boldsymbol{\bar{\omega}}) $ to zero, which are $ \frac{\partial \mathcal{L}(\boldsymbol{\tilde{\omega}}, \boldsymbol{\bar{\omega}})}{\partial \tilde{\omega}_i} = 0 $ and $ \frac{\partial \mathcal{L}(\boldsymbol{\tilde{\omega}}, \boldsymbol{\bar{\omega}})}{\partial \bar{\omega}_i} = 0 $, the following equations can be derived.
\begin{align}
\tilde{\omega}_i &= \left[ \frac{\tilde{\lambda}_i \delta_i f(|\mathcal{D}_i|)}{\tilde{\mu} R_i} \right]^{1/2}, \label{uplink_1_derive} \\
\bar{\omega}_i &= \left[ \frac{\bar{\lambda}_i f(|\mathbf{w}_i|)}{\bar{\mu} R_i} \right]^{1/2}. \label{uplink_2_derive}
\end{align}
As defined in \eqref{uplink_1_derive} and \eqref{uplink_2_derive}, if $ \tilde{\lambda}_i $ and $ \bar{\lambda}_i $ are zero, $ \tilde{\omega}_i $ and $ \bar{\omega}_i $ are zero. The uplink resource allocation $ \tilde{\omega}_i $ cannot be zero since user $ i $ needs to offload $ \delta_i $ portion of the local dataset unless $ \delta_i = 0 $. The uplink allocation for the weight transmission $ \bar{\omega}_i $ cannot be zero since user $ i $ always need to upload its weight parameter. Thus, $ \tilde{\lambda}_i \neq 0 $ and $ \bar{\lambda} \neq 0 $.
In case of $ \tilde{\mu} $ and $ \bar{\mu} $, the uplink resource allocations, $ \tilde{\omega}_i $ and $ \bar{\omega}_i $, are undefined if $ \tilde{\mu} $ and $ \bar{\mu} $ are zero, respectively.
Thus, $ \tilde{\mu} \neq 0 $ and $ \bar{\mu} \neq 0 $.
According to the KKT conditions, if $ \tilde{\mu} \neq 0 $ and $ \bar{\mu} \neq 0 $, we have the following conditions.
\begin{align}
\sum_{i \in \mathcal{I}} \tilde{\omega}_i &= 1, \label{uplink_1_con} \\
\sum_{i \in \mathcal{I}} \bar{\omega}_i &= 1. \label{uplink_2_con}
\end{align}

By substituting \eqref{uplink_1_derive} and \eqref{uplink_2_derive} into \eqref{uplink_1_con} and \eqref{uplink_2_con}, the values of $ \tilde{\mu} $ and $ \bar{\mu} $ are derived as follows.
\begin{align*}
\tilde{\mu} &= \left[ \sum_{i \in \mathcal{I}} \left[ \frac{\tilde{\lambda}_i \delta_i f(|\mathcal{D}_i|)}{R_i} \right]^{1/2} \right]^2, \\
\bar{\mu} &= \left[ \sum_{i \in \mathcal{I}} \left[ \frac{\bar{\lambda}_i f(|\mathbf{w}_i|)}{R_i} \right]^{1/2} \right]^2.
\end{align*}
Thus, the optimal values for the uplink bandwidth allocations can be derived as follows.
\begin{align*}
\tilde{\omega}_i &= \frac{1}{\sum_{i \in \mathcal{I}} \left[ \frac{\tilde{\lambda}_i \delta_i f(|\mathcal{D}_i|)}{R_i} \right]^{1/2}} \left[ \frac{\tilde{\lambda}_i \delta_i f(|\mathcal{D}_i|)}{R_i} \right]^{1/2}, \\
\bar{\omega}_i &= \frac{1}{\sum_{i \in \mathcal{I}} \left[ \frac{\bar{\lambda}_i f(|\mathbf{w}_i|)}{R_i} \right]^{1/2}} \left[ \frac{\bar{\lambda}_i f(|\mathbf{w}_i|)}{R_i} \right]^{1/2},
\end{align*}
which can be depicted as the proportional resource allocation.

Regarding the KKT conditions, if $ \tilde{\lambda}_i \neq 0 $ and $ \bar{\lambda}_i \neq 0 $, we have the following conditions.
\begin{align*}
t_i^{\text{edge}} = \Delta t, \\
t_i^{\text{local}} = \Delta t.
\end{align*}
From the above conditions, we have $ t_i^{\text{edge}} = t_i^{\text{local}} $ which can be expanded as
\begin{equation*}
\frac{\delta_i f(|\mathcal{D}_i|)}{\tilde{\omega}_i R_i} + \frac{\sum_{i \in \mathcal{I}} \delta_i f(|\mathcal{D}_i|) \tau}{\Gamma_E} = \frac{(1 - \delta_i) f(|\mathcal{D}_i|) \tau}{\gamma_i \Gamma_i} + \frac{f(|\mathbf{w}_i|)}{\bar{\omega}_i R_i}.
\end{equation*}
Thus, we can derive $ \delta_i $ as follows.
\begin{align*}
\delta_i = & \left[ \frac{f(|\mathcal{D}_i|)}{\tilde{\omega}_i R_i} + \frac{f(|\mathcal{D}_i \tau|)}{\Gamma_E} + \frac{f(|\mathcal{D}_i|) \tau}{\gamma_i \Gamma_i} \right]^{-1} \nonumber \\ & \left[ \frac{f(|\mathcal{D}_i|) \tau}{\gamma_i \Gamma_i} + \frac{f(|\mathbf{w}_i|)}{\bar{\omega}_i R_i} - \frac{\sum_{j \in \mathcal{I}, j \neq i} \delta_j f(|\mathcal{D}_j|) \tau}{\Gamma_E} \right].
\end{align*}
The dataset offloading, $ \delta_i $, is determined to balance off the time consumption among the local training and edge training.

\bibliographystyle{IEEEtran}
\bibliography{fl_mec_ref}

\end{document}